\documentclass[sigplan, nonacm]{acmart}

\makeatletter
\def\@ACM@checkaffil{
    \if@ACM@instpresent\else
    \ClassWarningNoLine{\@classname}{No institution present for an affiliation}%
    \fi
    \if@ACM@citypresent\else
    \ClassWarningNoLine{\@classname}{No city present for an affiliation}%
    \fi
    \if@ACM@countrypresent\else
        \ClassWarningNoLine{\@classname}{No country present for an affiliation}%
    \fi
}
\makeatother
\usepackage{amsmath}
\usepackage{graphicx}
\usepackage{amsmath,amsfonts}
\usepackage{amsthm}
\usepackage{balance}
\usepackage{xcolor}
\usepackage{fancyvrb}
\usepackage{algorithm}
\usepackage{algpseudocode}
\usepackage{subcaption}
\usepackage{xspace}
\usepackage{paralist}
\usepackage{url}
\usepackage{pgfplots}
\usepackage{pgfplotstable}
\pgfplotsset{compat=newest}
\usepackage{tikz}

\usepackage{bm}
\usepackage{wrapfig}
\usepackage{inputenc}  
\usepackage{lipsum}  
\usepackage{caption}
\usepackage{indentfirst}
\usepackage[htt]{hyphenat} 
\usepackage[hang,flushmargin]{footmisc}  
\usepackage{mathtools}
\usepackage{bold-extra}
\usepackage[all]{nowidow}
\usepackage{booktabs,multirow}

\newtheorem{theorem}{Theorem}
\newtheorem{definition}{Definition}
\newtheorem{setup}{Setup}

\newtheorem{lemma}{Lemma}

\newtheorem{proposition}{Proposition}

\newtheorem{terminology}{Terminology}

\newcommand{\ignore}[1]{}

\newcommand{\mycomment}[1]{}

\DeclareMathOperator*{\argmin}{argmin}

\newcommand{\ultraverse}{\textsf{Ultraverse}\xspace}

\newcommand{\para}[1]{\vspace{0.05in}\noindent{\bf{#1}}}

\newcommand\vldbdoi{XX.XX/XXX.XX}

\input{savespace.tex}

\usepackage{hyperref}

\pdfoutput=1
\begin{document}

\makeatletter
\everypar={\looseness=-1}
\makeatother

\title{\ultraverse: An Efficient \textit{What-if} Analysis Framework for Software Applications Interacting with Database Systems}

\author{Ronny Ko}
\affiliation{\institution{Osaka University}\country{Japan}, \institution{Ohio State University}\country{USA}, \institution{Sensor Tech}\country{Republic of Korea}}
\email{hrko@g.harvard.edu}

\author{Chuan Xiao}
\affiliation{\institution{Osaka University}\country{Japan}, \institution{Nagoya University}\country{Japan}}
\email{chuanx@ist.osaka-u.ac.jp}

\author{Makoto Onizuka}
\affiliation{\institution{Osaka University}\country{Japan}}
\email{onizuka@ist.osaka-u.ac.jp}

\author{Zhiqiang Lin}
\affiliation{\institution{Ohio State University}\country{USA}}
\email{zlin@cse.ohio-state.edu}

\author{Yihe Huang}
\affiliation{\institution{Databricks}\country{USA}}
\email{yihehuang@g.harvard.edu}



\begin{abstract}
Existing \textit{what-if} analysis systems are predominantly tailored to operate on either only the application layer or only the database layer of software. This isolated approach limits their effectiveness in scenarios where intensive interaction between applications and database systems occurs. To address this gap, we introduce \ultraverse, a \textit{what-if} analysis framework that seamlessly integrates both application and database layers. \ultraverse employs dynamic symbolic execution to effectively translate application code into compact SQL procedure representations, thereby synchronizing application semantics at both SQL and application levels during \textit{what-if} replays. A novel aspect of \ultraverse is its use of advanced query dependency analysis, which serves two key purposes: (1) it eliminates the need to replay irrelevant transactions that do not influence the outcome, and (2) it facilitates parallel replay of mutually independent transactions, significantly enhancing the analysis efficiency. \ultraverse is applicable to existing unmodified database systems and legacy application codes. Our extensive evaluations of the framework have demonstrated remarkable improvements in \textit{what-if} analysis speed, achieving performance gains ranging from 7.7x to 291x across diverse benchmarks. 
\end{abstract}

\maketitle



\section{Introduction}
\label{sec:introduction}

\textit{What-if} analysis in computer systems is a pivotal technique for simulating hypothetical modifications in data records or transaction histories, aiming to predict outcomes under altered conditions. 
This allows individuals and organizations to explore different scenarios and understand potential impacts before making decisions (e.g., to understand how changes in strategy could affect market share, sales, and customer behavior). \textit{What-if} analysis is particularly useful in business intelligence~\cite{forbes}, and is also available in commodity tools such as Microsoft Excel~\cite{what-if-microsoft}, Power BI~\cite{what-if-power-bi}, Tableau~\cite{what-if-tableau}, and Oracle Analytics~\cite{what-if-oracle}.

However, the existing tools have fundamental limitations, because they only focus on individual data records, failing to capture the comprehensive data flows across the database and applications. This is a problem, because a majority of real-life software applications involve \textit{both} database \& application data flows. Ignoring these factors during a \textit{what-if} analysis will result in semantically incorrect results from the application's viewpoint. Within this context, designing \textit{what-if} analysis in application software encounters two significant challenges: preserving the correctness of the outcomes and ensuring the efficiency of computing the outcome.



\para{Correctness Issue}: Today's most software is implemented using application languages with databases as persistent storage. Therefore, it is crucial for \textit{what-if} analysis to consider data flows between application code and database's state. 
\autoref{fig:new-follows-code} illustrates such an example using JavaScript with MySQL. This example implements an e-commerce web server's request handler function, where the \texttt{NewOrder} function serves a user request of placing a new order. This application function is essentially an \textit{application-level} transaction internally executing two SQL queries. The first \texttt{SELECT} query verifies if the requesting user has a registered shipping address in the \texttt{Address} table. If yes, the second \texttt{INSERT} query places the actual order by updating the \texttt{Orders} table. If the user does not have an address, the function returns an error. Given this request handler setup, suppose that user Alice had a registered address and placed an order successfully, and later we want to conduct a \textit{what-if} analysis of an alternate scenario where Alice had attempted to place an order without a registered address. 
According to \autoref{fig:new-follows-code}'s web application code, if Alice had a registered address, the \texttt{NewOrder} call would traverse the \texttt{true} branch of the \texttt{if} statement as the \texttt{SELECT} query would return an address record. On the other hand, during the \textit{what-if} analysis where Alice hadn't had a registered address, the \texttt{NewOrder} call should instead traverse the \texttt{false} branch as the \texttt{SELECT} query would return no address record. To correctly simulate this \textit{what-if} analysis scenario based on the changed past (i.e., Alice does not have an address), we should first modify the \texttt{Address} table's past state such that it did not store Alice's address, and then replay the \texttt{NewOrder} call. During this, notice that simply replaying only the individual DBMS queries (i.e., \texttt{SELECT} and \texttt{INSERT}) with a modified \texttt{Address} table would not give a correct analysis result, because DBMS queries are unaware of application-level semantics, which means the DBMS won't interpret the return value of the \texttt{SELECT} query as a trigger of the application-level \texttt{false} branch that does not replay the \texttt{INSERT} query.
A number of state-of-the-art DBMS-oriented \textit{what-if} analysis techniques 
(\cite{mahif, caravan, tiresias, qfix, whatif1, whatif2, whatif3}) have been proposed to efficiently answer \textit{what-if} database queries. However, they all fail to serve application-level \textit{what-if} analyses, because they are designed only for database queries and query logs, not capturing the higher-level application code's control flow logic or data flows that can impact the database's state. As a result, their \textit{what-if} analysis leads to inaccurate results from the application's viewpoint. \looseness=-1


\begin{figure}
\begin{footnotesize}
\begin{Verbatim}[frame=single, baselinestretch=0.7, commandchars=\\\{\}]
\textcolor{blue}{{1: }}\textbf{function} NewOrder \textbf{(}orderer_uid, order_id\textbf{):} 
\textcolor{blue}{{2: }}   \textbf{var} result_rows \textbf{=} \textbf{SQL_exec(}\textit{\textcolor{brown}{"SELECT COUNT(*)  FROM Address}}
\textcolor{blue}{{3: }}    \textit{\textcolor{brown}{ WHERE owner_id = "}} \textbf{+} orderer_uid);
\textcolor{blue}{{4: }}   \textbf{if (}result_rows[0]['COUNT(*)'] \textbf{!=} 0\textbf{)}
\textcolor{blue}{{5: }}      \textbf{SQL_exec(}\textcolor{brown}{"INSERT INTO Orders (oid, ord_uid) VALUES ("}
\textcolor{blue}{{6: }}        + order_id + \textcolor{brown}{", "} + orderer_uid + \textcolor{brown}{")"}\textbf{);}
\textcolor{blue}{{7: }}   \textbf{else} 
\textcolor{blue}{{8: }}      \textbf{return} \textit{\textcolor{brown}{"Error: User "}} + orderer_uid + \textit{\textcolor{brown}{" has no address"}}\textbf{;}
\end{Verbatim}
\end{footnotesize}
\caption{An E-commerce web server's user request handler.}
\label{fig:new-follows-code}
\end{figure}

\para{Efficiency Issue}: Some research works focusing on application attack recovery and forensics~\cite{warp, rail, chromepic, regexnet} manage to avoid the consistency issue by replaying both the application code and their associated DBMS queries. However, these approaches often lack scalability. For instance, given that each headless (graphic-less) Chromium browser instance requires around 100MB of RAM, the \textit{what-if} analysis overhead for a large user base (say, 1 million) becomes practically infeasible. 
Such scalability is also a prominent issue in \textit{what-if} DBMSes. Mahif~\cite{mahif}, the state-of-the-art DBMS for historical \textit{what-if} queries of the four basic types (\texttt{INSERT}, \texttt{UPDATE}, \texttt{DELETE}, and \texttt{SELECT}), exhibits this problem. The system's resource consumption increases exponentially with the linear growth of the transaction history size, making it incapable of processing a transaction history that contains more than even 2,000 SQL queries (\autoref{subsec:comparison}). Further, Mahif fails to capture complex semantics such as SQL-level \texttt{TRANSACTION}s or application-level transactions. \looseness=-1

To address the requirements of correctness and efficiency in application-wide \textit{what-if} analysis, we introduce \ultraverse, a novel framework for application-level \textit{what-if} analysis. Conventionally, \textit{what-if} analysis is akin to a ``retroactive operation" that modifies a \textit{past} operation in a series of committed ones, such as cancelling a previously committed insertion of `5' into a queue at time \textit{t=3}. 
\ultraverse stands out as the first DBMS to efficiently support retroactive operations on both SQL queries and application-level transactions while preserving the correctness of both SQL-level and application-level semantics. Our framework consists of two components: \textbf{(1) An SQL transpiler} transforms application code into its semantically equivalent SQL procedures; and \textbf{(2) A retroactive operation DBMS plugin} efficiently executes \textit{what-if} analysis on the transpiled SQL procedures. \looseness=-1

    \para{SQL Transpiler (\autoref{sec:transpiler})} translates each application-level transaction (developed in application languages) into a compact SQL \texttt{PROCEDURE} such that executing the SQL procedure has the same effect to the persistent database as executing the original application-level transaction.
    During retroactive operation, \ultraverse replays these transpiled \texttt{PROCEDURE}s rather than the original application code, ensuring both a correct (i.e., transactionally atomoic) and fast replay. This replay is fast because the transpiler prunes application logic that doesn't affect persistent storage. Additionally, the transpiler consolidates multiple SQL queries belonging to the same application-level transaction into a single SQL \texttt{PROCEDURE}, thereby reducing round-trip communications between the DBMS client and server. In the \autoref{fig:new-follows-code} example, the transpiled code would include the \texttt{SELECT} and \texttt{INSERT} queries within one \texttt{PROCEDURE}, reducing the round-trip time (RTT) from two to one. This reduction becomes even more significant when an application-level transaction uses loops to execute a larger number of individual SQL queries, as all these queries will execute within a single RTT as part of the single \texttt{PROCEDURE} call. \ultraverse's SQL transpiler is compatible with legacy application code and doesn't require manual modifications by developers. However, it faces the significant challenge of transpiling modern application languages (like JavaScript, PHP) having dynamic features that are difficult to analyze statically, such as dynamic variable type resolution, modifiable function pointers as object properties, asynchronous event calls, etc. Converting such dynamic behaviors into semantically equivalent SQL code is non-trivial. To address this, \ultraverse's transpiler utilizes abstract syntax trees and dynamic symbolic execution (\autoref{subsec:dse}). \looseness=-1

    \para{Retroactive Operation DBMS Plugin (\autoref{sec:design})} reduces the number of queries to be replayed in order to achieve high efficiency during \textit{what-if} analysis. For this, we introduce a fine-grained query dependency analysis to exclude queries whose execution results are deterministically unaffected by the modified \textit{what-if} target query (or queries). In detail, we perform query dependency analysis on the SQL \texttt{PROCEDURE} code produced by the SQL transpiler. This analysis occurs in two orthogonal aspects: column-wise (\autoref{subsec:columnwise-dependency-analysis}) and row-wise (\autoref{subsec:rowwise-analysis}) analysis, which respectively consider all the table columns and rows (i.e., cells) each query could potentially read or write.     
    Additionally, \ultraverse safely replays multiple queries simultaneously if their read/write sets (generated by the query dependency analysis) have no conflicts, further accelerating the replay. Such parallel execution (enabled by the the query dependency analysis) is a powerful enabler of multi-core OS resources. \ultraverse's hash-jumper (\autoref{subsec:hashjump}) enables further speedup by detecting the opportunities to safely early-stop the retroactive replay without harming correctness of the result. \ultraverse's retroactive operation plugin is applicable to unmodified DBMSes that support SQL, enhancing its practical deployability. \ultraverse supports rich SQL semantics\footnote{\ultraverse supports DDL (schema updates), transactions, stored procedures, triggers, control flow branches, temporary variables, and native APIs.}(\autoref{subsec:columnwise-dependency-analysis}), which are mostly unsupported by prior \textit{what-if} DBMSes. \looseness=-1

We implemented \ultraverse and evaluated its \textit{what-if} analysis performance against a baseline application built with NodeJS + MySQL. Across various BenchBase benchmarks \cite{oltp-bench} (TPC-C, TATP, Epinions, and SEATS), \ultraverse exhibited a speedup ranging from 7.7x to 512x (\autoref{subsec:performance}), and was found to be 6450x faster than Mahif \cite{mahif} (\autoref{subsec:comparison}). We also tested \ultraverse on an open-source e-commerce web application (AStore \cite{astore})\ignore{ and machine-learning data analytics (Apache Hivemall \cite{hivemall})}, with \ultraverse achieving a \textit{what-if} analysis speedup between 18.7x and 601x. The additional space overhead for \textit{what-if} analysis using \ultraverse is marginal, e.g., 12 to 110 bytes for the log size per query, in contrast to hundreds of bytes for MySQL's query log.  \looseness=-1

\para{Contributions}. Our work makes the following contributions: 
\begin{itemize}
\item We design an efficient and correct \textit{what-if} analysis framework for database-using applications.
\item We design the novel query dependency analysis for efficient retroactive operation for database systems. 
\item We implement and demonstrate the efficiency of the \ultraverse prototype through an empirical evaluation.
\end{itemize}

\section{System Overview}
\label{sec:overview}

\begin{figure}[!t]
\includegraphics[width=.5\textwidth]{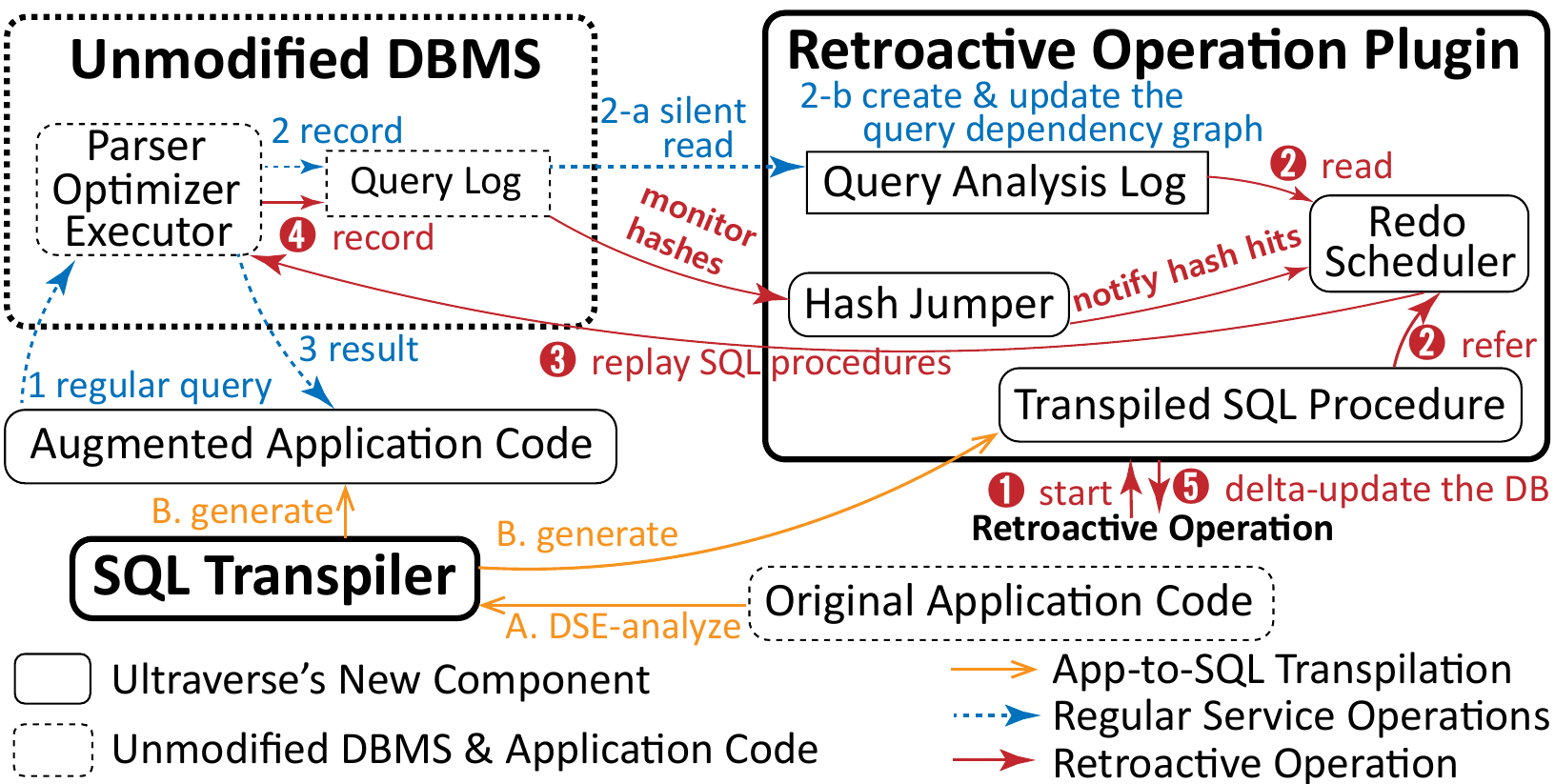}
  \caption{The comprehensive architecture of \ultraverse.}
  \label{fig:system-overview}
\end{figure}

\autoref{fig:system-overview} shows \ultraverse's comprehensive system architecture, comprising the SQL transpiler and the retroactive DBMS plugin. The SQL transpiler accepts the original application code as input and generates two components as outputs: (1) augmented application-level transaction code (\autoref{fig:new-follows-transpiled-javascript}); and (2) its semantically equivalent SQL \texttt{PROCEDURE} (\autoref{fig:new-follows-transpiled-sql}). The augmented application-level transaction is executed during the application's regular service operations, during which it internally sends SQL queries to the unmodified DBMS as usual, which in turn executes the queries and logs them in the query log. At the same time, the augmented application code also asynchronously records the information of the invoked application-level transaction (line 2 in \autoref{fig:new-follows-transpiled-javascript}), so that \ultraverse's query analyzer can include this transaction in the query dependency analysis. \ultraverse's smallest granularity of query dependency analysis is the table columns and rows (i.e., cells) that each query potentially accesses. As the query analyzer reads the query log, it can deduce read-write dependencies among the queries as well as among the application-level transactions.  \looseness=-1
When a user makes a request to retroactively add, remove, or modify past queries or application-level transactions, \ultraverse's query analyzer examines the query dependency log and identifies only those queries/transactions that need to be rolled back and replayed. During the replay phase, \ultraverse executes the equivalent SQL \texttt{PROCEDURE} of each targeted application-level transaction, and executes multiple mutually independent (i.e., having no read/write dependencies) \texttt{PROCEDURE}s in parallel to improve the speed, while still ensuring a strongly serialized final database state-- as if all queries were committed in the same sequential order as the original execution. \looseness=-1

\section{SQL Transpiler}
\label{sec:transpiler}

\begin{figure}[t!]
\begin{footnotesize}
\begin{Verbatim}[frame=single, baselinestretch=0.7, commandchars=\\\{\}]
\textcolor{blue}{{1:}}\textbf{function} NewOrder\textbf{(}orderer_uid, order_id\textbf{)}
\textcolor{blue}{{2:}}\{  \textbf{Ultraverse_log(}\textcolor{brown}{`function NewOrder(}$\{orderer_uid\}\textcolor{brown}{,} 
\textcolor{blue}{{3:}}     $\{order_id\}\textcolor{brown}{)`}\textbf{)};
\textcolor{blue}{{4:}}   ...  \textcolor{red}{// the same function body of NewOrder() in Figure 1}
\textcolor{blue}{{5:}}\}
\end{Verbatim}
\end{footnotesize}
\caption{Augmented JavaScript code of \autoref{fig:new-follows-code}.}
\label{fig:new-follows-transpiled-javascript}
\end{figure}

\begin{figure}[t!]
\begin{footnotesize}
\begin{Verbatim}[frame=single, baselinestretch=0.7, commandchars=\\\{\}]
{\textbf{DECLARE PROCEDURE} NewOrder\textbf{(IN} orderer_uid \textbf{INT}, }
{\textbf{IN} order_id \textbf{VARCHAR(8))} NewOrder_Label\textbf{:BEGIN}}
       {\textbf{DECLARE} sql_out1_0_count \textbf{INT};}
       {\textbf{SELECT COUNT(*) INTO} sql_out1_0_count \textbf{FROM} Address}
         {\textbf{WHERE} owner_uid = orderer_uid;}
       {\textbf{IF (}sql_out1_0_count != 0\textbf{) THEN}}
          {\textbf{INSERT INTO} Orders \textbf{VALUES (}orderer_uid, order_id\textbf{)};}
       {\textbf{ELSE}}
          {\textbf{SELECT CONCAT(}\textcolor{brown}{"Error: User "}, orderer_uid,} 
          { \textcolor{brown}{" has no addresss"}\textbf{)};}
          {\textbf{LEAVE} NewOrder_Label;}
       {\textbf{END IF};}
\end{Verbatim}
\end{footnotesize}
\caption{Transpiled \texttt{SQL} \texttt{PROCEDURE} version of \autoref{fig:new-follows-code}.}
\label{fig:new-follows-transpiled-sql}
\end{figure}


\ultraverse's SQL transpiler is designed to analyze application code and identify the set of SQL queries that belong to the same application-level transaction, in order to atomically replay those queries as well as their wrapping application code as a single transaction. Then, the transpiler generates \textit{augmented} application code (as shown in \autoref{fig:new-follows-transpiled-javascript}) that asynchronously records (line 2) which application-level transaction is called in runtime. The transpiler also generates an SQL \texttt{PROCEDURE} equivalency of each application-level transaction (as shown in \autoref{fig:new-follows-transpiled-sql}), which captures all (and only those) data flows within the application-level transaction affects the database state. Later, these SQL \texttt{PROCEDURE}s are executed by the retroactive DBMS plugin (\autoref{sec:design}) during the retroactive replay. To transpile the application code and its associated SQL queries into such SQL \texttt{PROCEDURE}s, \ultraverse's approach is to analyze application code not only statically with abstract syntax trees (AST), but also dynamically based on dynamic symbolic execution (DSE).

\subsection{Dynamic Code Analysis}
\label{subsec:dse}

Static code analysis is a code scanning technique to examine the behaviors of a target program and its data flows without executing the program, but it has limitations in capturing dynamic runtime behaviors of modern application programming languages~\cite{javascript}. In case of JavaScript, variables has no static types and are resolved dynamically in runtime. Control flow targets also can be resolved dynamically (e.g., \texttt{var} \texttt{result = myObject.methodName(arg1, arg2)}, where the value of \texttt{.methodName} is dynamically assigned). Furthermore, some APIs are a blackbox involving undeterminism (e.g., \texttt{rand()} or a response from a remote endpoint).  
It is difficult for static analysis to correctly transpile such dynamic behaviors into a statically typed language (e.g., SQL), because their dynamic types, control flow targets, or return values are concretized only in runtime.


To address the challenge of dynamism, \ultraverse dynamically analyzes the application code based on dynamic symbolic execution (DSE). Symbolic execution is originally a technique to discover  viable execution paths of an application program and generate an inventory of program inputs that lead to each path. Symbolic execution is commonly employed for software debugging and malware analysis. To maintain efficiency, our research strategically uses DSE to traverse only truly reachable execution paths in real-time, effectively pruning a vast number of unreachable, hypothetical execution paths.
\looseness=-1


The underlying concept of \ultraverse involves executing the target application based on DSE and projecting all dynamic runtime behaviors of the application code into a concretized execution path tree. This tree serves as a flowchart-like depiction of all potential program execution scenarios \textcolor{black}{(i.e., all control flow graphs)}. \ultraverse adopts a \textit{concolic} (i.e., concrete + symbolic) strategy for symbolic execution~\cite{ball2015deconstructing}, fully executing each testcase before exploring alternate paths. This method proves more practical than a purely static program analysis, which often identifies an impractical number of unreachable alternate control flows. \looseness=-1

DSE is dynamic in nature. Each testcase is executed by an actual runtime programming language interpreter (e.g., JavaScript engine or Python interpreter), ensuring each execution is a real-time action and does not produce false positives (i.e., it explores only reachable paths). This approach is highly compatible with dynamic programming languages such as JavaScript, where dynamic code evaluation and complex type-coercion often stymie effective static analysis. \ultraverse regards the input parameters of application-level transactions and the return values of database API calls as symbolic, and expresses the values of the variables storing these data flows as symbolic expressions. The runtime values of all other variables are considered as concretized constants to ensure deterministic execution. \looseness=-1

\subsection{Procedure of Transpilation}
\label{subsec:transpilation}

\ultraverse analyzes application code with DSE to convert its each transaction into a semantically equivalent SQL \texttt{PROCEDURE}. This transpilation process involves three steps: (1) application code instrumentation; (2) execution of the instrumented application; and (3) Z3-to-SQL transpilation.  \looseness=-1

\begin{figure}[t!]
  \includegraphics[width=0.5\textwidth]{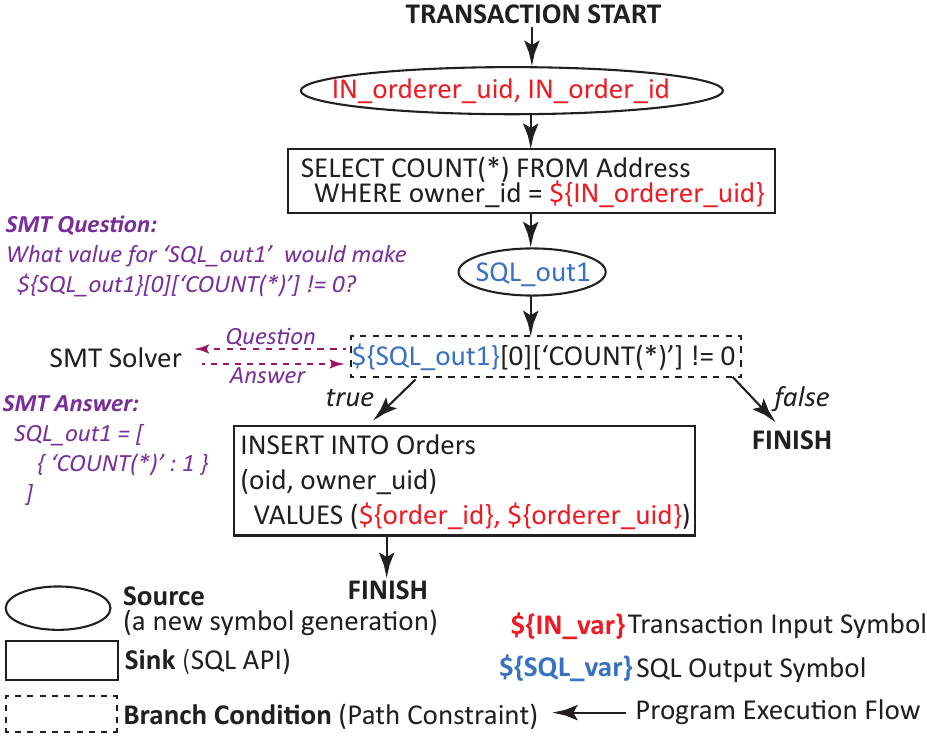}
  \caption{DSE's execution path tree of the \texttt{"NewOrder"} application-level transaction in \autoref{fig:new-follows-code}.}
  \label{fig:dse-example}
\end{figure}

\para{Step 1: Application Code Instrumentation}. 
This process creates an AST of the target application and installs a hook function at every operation of an application’s code (e.g., binary operations, function entry/exit, conditional branches, or read/write of variables). After then, as the instrumented application code executes via an unmodified runtime language interpreter, the injected hook silently builds symbolic expressions of all (intermediate) variables and branch conditions at every dynamic operation. Symbols in \ultraverse's DSE analysis are defined as follows: (1) input parameters to each application-level transaction; (2) return values of database API calls; and (3) return values of all nondeterministic native APIs (e.g., \texttt{Math.random()}, \texttt{Date.gettime()}). These symbols are considered as blackboxes that contain unpredictable values at each run. 
In \autoref{fig:new-follows-code}, the transaction input symbols are \texttt{src\_user} and \texttt{dst\_user}; the database's return value symbol is \texttt{result\_rows}. The hooks render symbolic expressions in the Z3 script language to be used by the Z3 Theorem Prover, which helps discover concrete values for symbols (i.e., program inputs) leading to each branch condition in the program. \looseness=-1

To minimize the overhead of instrumentation and analysis time, \ultraverse instruments only the application-level transaction functions, which can be auto-detected by \ultraverse or manually specified by the user. For example, if the application uses the ExpressJS web application framework, \ultraverse regards the framework's \texttt{router} event calls that respond to a user's HTTP request as application-level transaction calls. All nested functions called by the top-level functions are instrumented as well.  \looseness=-1

\para{Step 2: Execution of Instrumented Application.}
\ultraverse runs the instrumented application generated at Step 1. Our scope of symbolic execution is the application-level transaction function (from the function entry to its exit). In the initial execution, the program is executed with a randomized initial symbolic value set as a seed testcase. The execution starts with the seed values as input, and at each database API call, the instrumented hook bypasses DBMS access (treating its internal logic as a blackbox) and returns a pre-set symbolic value instead. As the execution finishes, the program trace is collected and a path condition is developed. This condition contains each of the logical constraints on symbols across the program's control flow. If the SMT solver determines that there is a feasible symbolic value set which would have led the transaction differently at each conditional branch, such new transaction inputs and return values of database API calls are derived and added to the testcase queue. This process repeats until either there are no more new execution paths or a user-set analysis timeout is reached. The final result is a comprehensive program execution path tree consisting of only truly reachable execution paths and a list of symbolic input values leading to each path. Each path in the tree compactly captures only those data flows affecting the database's state. \looseness=-1

\autoref{fig:dse-example} illustrates the \texttt{NewOrder} application-level transaction's (\autoref{fig:new-follows-code}) execution path tree, which covers all possible execution paths of the \texttt{NewOrder} transaction.
At this point, we have successfully transformed the application code subject to dynamism into a concretized program execution tree. \looseness=-1

\para{Step 3: Z3-to-SQL Transpilation.} In this final step, \ultraverse's transpiler converts the Z3 symbolic execution tree into a logically equivalent SQL \texttt{PROCEDURE}. Each Z3 operator in the program execution tree is mapped to its semantically equivalent SQL operator or API (e.g., Z3’s \texttt{str.++(a, b)} operator is equivalent to SQL’s \texttt{CONCAT(a, b)} API). The end result is a modified application-level transaction which internally calls a newly generated SQL \texttt{PROCEDURE}, running the SQL code equivalent to only the application logic contributing to directly/indirectly updating the database. \autoref{fig:new-follows-transpiled-sql} is the transpiled SQL version of \autoref{fig:new-follows-code}. While \autoref{fig:new-follows-transpiled-sql} is logically equivalent to \autoref{fig:new-follows-code}, its processing time is shorter because it costs 1 RTT between the DBMS server and client instead of 2 RTTS (by executing \texttt{SELECT} and \texttt{INSERT} within a single \texttt{PROCEDURE} call instead of executing \texttt{SELECT} and \texttt{INSERT} separately).  \looseness=-1

\subsection{Handling Practical Challenges} 
\label{subsec:dse-discussion}

Code analysis by concolic execution is often insufficient to capture all possible execution behaviors of a target program. We discuss these issues and describe how we address them. \looseness=-1

\para{Handling Path Explosion: } DSE analysis can face a path explosion issue under certain conditions. Specifically, if one symbol is used as a loop condition, the number of dynamically explorable paths can infinitely grow. \ultraverse detects such potential path explosions based on path generation patterns in the path tree. Specifically, \ultraverse uses the information of AST and the loop summarization technique~\cite{10.1145/2001420.2001424, 8241837} to detect loop patterns. \looseness=-1

\para{Handling Unreached Path: } Due to limitations of computational resources, the SMT solver may fail to resolve some symbolic constraints of branch conditions before the timeout. For example, a branch condition may involve complex regular expressions of string variables. The execution paths with constraints that the SMT solver fails to solve may or may not be reachable in actual runtime. For such an uncertainly unexplored branch path in the final transpiled \texttt{PROCEDURE}, \ultraverse inserts an \texttt{SQLSTATE} command with a reserved signal number. If this signal is triggered during regular service or \textit{what-if} replay, \ultraverse learns about the symbolic values that led to this path, runs a delta DSE analysis on this path, and applies a delta update to the transpiled \texttt{PROCEDURE} to incorporate this path.  \looseness=-1

\para{Blackbox APIs:} Some JavaScript native APIs are difficult to model into Z3. Such examples include arbitrary dynamic code execution (i.e., \texttt{eval()}), non-deterministics functions (e.g., \texttt{Math.random()}), or network APIs communicating with unknown external endpoints. Whenever such unpredictable APIs are called during symbolic execution, \ultraverse treats them as blackbox APIs, and spawns \& links a new symbol to their return values, which is similar to spawning an SQL symbol (e.g., \texttt{SQL\_out1} in \autoref{fig:dse-example}) upon each DBMS API call. And later, during the \textit{what-if} replay simulation based on the transpiled SQL procedure, when it encounters such a blackbox API call, \ultraverse provides two options: (1) the user (i.e., \textit{what-if} analyst) can handcode the return value of each unpredictable API call as part of \textit{what-if} analytics; or (2) \ultraverse computes the return value of the blackbox API call by executing the actual JavaScript native API (e.g., \texttt{eval()} or \texttt{Math.random()}) on the fly. In the latter case, \ultraverse's transpiled SQL procedure is instrumented in such a way that it temporarily returns from the SQL procedure, runs the JavaScript native API call, captures its return value, and resumes the SQL procedure from where it left off. \looseness=-1

\para{Server-Client Communication: } If the target application is a web browser-based service, it may contain client-side JavaScript logic that pre-processes or post-processes the input/return values of application-level transactions to be executed by the server-side code. If such client-side logic contains data flows affecting the server-side database's state, such client-side code should also be factored in to ensure a sound \textit{what-if} analysis. \ultraverse can be extended to handle such cases by running DSE on the client-side webpage as well. Specifically, it will treat a webpage's DOM nodes that store the client's inputs (e.g., \texttt{<input>}) as symbols and trace their evolving symbolic expressions throughout the webpage's logic, which includes the client-to-server communication via JavaScript's HTTP APIs (e.g., \texttt{XMLHTTPRequest.send()}). During the DSE analysis of the webpage, the client-side code's reserved variables that fingerprint the client-specific identity (e.g., \texttt{navigator.userAgent}) are also considered as client-side symbols. \looseness=-1

In \autoref{appendix:dynamic-application}\footnote{\label{footnote:appendix}The URL for the full paper with Appendix: \url{https://arxiv.org/abs/2211.05327}}, we provide examples of application-level dynamism and explain how \ultraverse's transpiled code handles them.

\subsection{Correctness of Transpilation}
\label{subsec:transpilation-correctness}

We define the correctness of transpilation as follows: a transpilation of an application-level transaction into an SQL \texttt{PROCEDURE} is correct if and only if the execution of both transactions (i.e., the original one and the transpiled one) always have the identical effect to their associated persistent SQL database for all possible sets of symbolic values belonging to the transaction (i.e., transactional input arguments, the database calls' return values, and non-deterministic API calls' return values).

As explained in \autoref{subsec:dse}, an application implemented by a dynamic programming language can trigger undeterministic behaviors of a program by: dynamically resolving variable types, dynamically resolving control flow targets, and using undeterministic return values of blackbox APIs. \ultraverse's transpiler handles dynamic type coercion  under control by dynamically monitoring and capturing each concretized type and applying the equivalent SQL variable type during SQL code conversion. \ultraverse further handles undeterministic blackbox APIs by using the technique of dynamic symbol spawning (\autoref{subsec:dse-discussion}). Lastly, \ultraverse handles \textit{symbolic} control flow targets based on the combination of the following two techniques: (1) identifies all available jump targets discovered during the DSE analysis and incorporates them into the final transpiled code with \textit{if}-statements conditioned on the symbolic jump target (i.e., the function name to be invoked); and (2) any other jump targets not identified during the DSE anlaysis (but exists) can be detected during the regular operations or retroactive operations, at which point the transpiled SQL code captures it on the fly and delta-updates the transpiled code by incorporating the newly discovered jump target. In \autoref{appendix:dynamic-application}, we explain how \ultraverse handles various types of dynamism with examples.

\ultraverse ensures that the application's all intermediate states that can occur during a retroactive operation are detected as one of those states either discovered during the offline DSE analysis or detected in real time. Based on such full capturing of all possible program states (either offline or online), \ultraverse preserves the correctness of retroactive operation.

Our \ultraverse prototype assumes that the application is written in a single programming language. Meanwhile, \ultraverse in principle supports multi-language applications  based on the same proposed techniques if a cross-language DSE engine is available. \ultraverse also assumes that the application exclusively uses SQL databases.

\section{Retroactive Database System}
\label{sec:design}




Once the SQL transpiler (\autoref{sec:transpiler}) has transpiled application-level transactions to semantically equivalent SQL \texttt{PROCEDURE}s, our next task is to analyze their read/write dependencies to enable efficient retroactive opeation, which will replay only those transactions dependent on the retroactive target transactions.


Consider a database $\mathbb{D}$, a set $\mathbb{Q}$ of queries $Q_i$ where $i$ represents the query's {commit order} (i.e., query index), and {$Q'_\tau$ is the target query to be retroactively added, removed, or changed at the commit order $\tau$ within $\{Q_1, Q_2, ... Q_\tau ... Q_{|\mathbb{Q}|}\} = \mathbb{Q}$}. 
In case of retroactively adding a new query $Q'_\tau$, $Q'_\tau$ is to be {inserted (i.e., executed)} right before $Q_\tau$. 
In case of retroactively removing the existing query $Q_\tau$ $(i.e., Q'_
\tau = Q_\tau)$, $Q_\tau$ is to be simply removed from the committed query list.
In case of retroactively {changing the existing query $Q_\tau$ to $Q'_\tau$, $Q_\tau$ is to be replaced by $Q'_\tau$}. 
The retroactive operation on the target query $Q'_\tau$ is equivalent to transforming $\mathbb{D}$ to a new database state that matches the one generated by the following procedure: 
\begin{enumerate}
    \item \textbf{Rollback Phase}: roll back $\mathbb{D}$'s state to commit index $\tau - 1$ by rolling back $Q_{|\mathbb{Q}|}, Q_{|\mathbb{Q}-1|}, \ldots Q_{\tau+1}, Q_\tau$.
    \item \textbf{Replay Phase}: do one of the following:
    \begin{itemize}
        \item To retroactively add $Q'_\tau$, execute $Q'_\tau$ and then replay $Q_\tau, \ldots Q_{|\mathbb{Q}|}$. 
        \item To retroactively remove $Q'_\tau$, replay $Q_{\tau+1}, \ldots Q_{|\mathbb{Q}|}$ (without $Q_\tau$). 
        \item To retroactively change $Q_\tau$ to $Q'_\tau$, execute $Q'_\tau$ and then replay $Q_{\tau+1}, \ldots Q_{|\mathbb{Q}|}$.  
    \end{itemize}
\end{enumerate}

Instead of exhaustively rolling back and replaying all $Q_\tau, \ldots Q_{|\mathbb{Q}|}$, \ultraverse introduces query dependency analysis to rollback and replay only those queries dependent on the retroactive target query $Q'_\tau$, speeding up the process while guaranteeing the expected same final state of $\mathbb{D}$.  \looseness=-1

\subsection{A Motivating Example}
\label{subsec:example}

\autoref{fig:dependency-graph-website-columnwise} depicts a comprehensive example of \ultraverse's query dependency graph for an e-commerce web service as an extension of \autoref{fig:new-follows-code}, which comprises four tables and one procedure. The \texttt{Users} table stores each user's id, nickname, and email. The \texttt{Address} table stores each user's registered shipping address. The \texttt{Orders} table records placed orders from users. The \texttt{Stats} table logs the trend of user purchases over time. The \texttt{NewOrder} procedure places a new order only if the intended user has a registered address. \looseness=-1

Initially, the service creates the \texttt{Users}, \texttt{Address}, \texttt{Orders}, and \texttt{Stats} tables (Q1$\sim$Q4) as well as the \texttt{NewOrder} procedure (Q5). Then, Alice signs up (Q6), registers a shipping address (Q7), and places an order (Q8). Bob also signs up (Q9) and places an order (Q10). After all these, the statistical data of user purchases gets recorded (Q11).

In \autoref{subsec:columnwise-dependency-analysis}, we will explain \ultraverse's column-wise dependency analysis technique (inspired by column-level data lineage~\cite{column-lineage}) that allows the \textit{what-if} analysis to only roll back Q8, Q10, and Q11, and then replay Q10 and Q11. In other words, we can avoid unnecessarily rolling back and replaying Q9, as the table columns Q9 reads (i.e., \texttt{Users.uid}) remain unaffected by the retroactive removal of Q8 that writes to table columns \texttt{Orders.*}. In \autoref{subsec:rowwise-analysis}, we will also explain \ultraverse's row-wise dependency optimization technique which further avoids the unnecessary rollback and replay of Q10.

\looseness=-1

\begin{figure}[t!]
  \includegraphics[width=.5\textwidth]{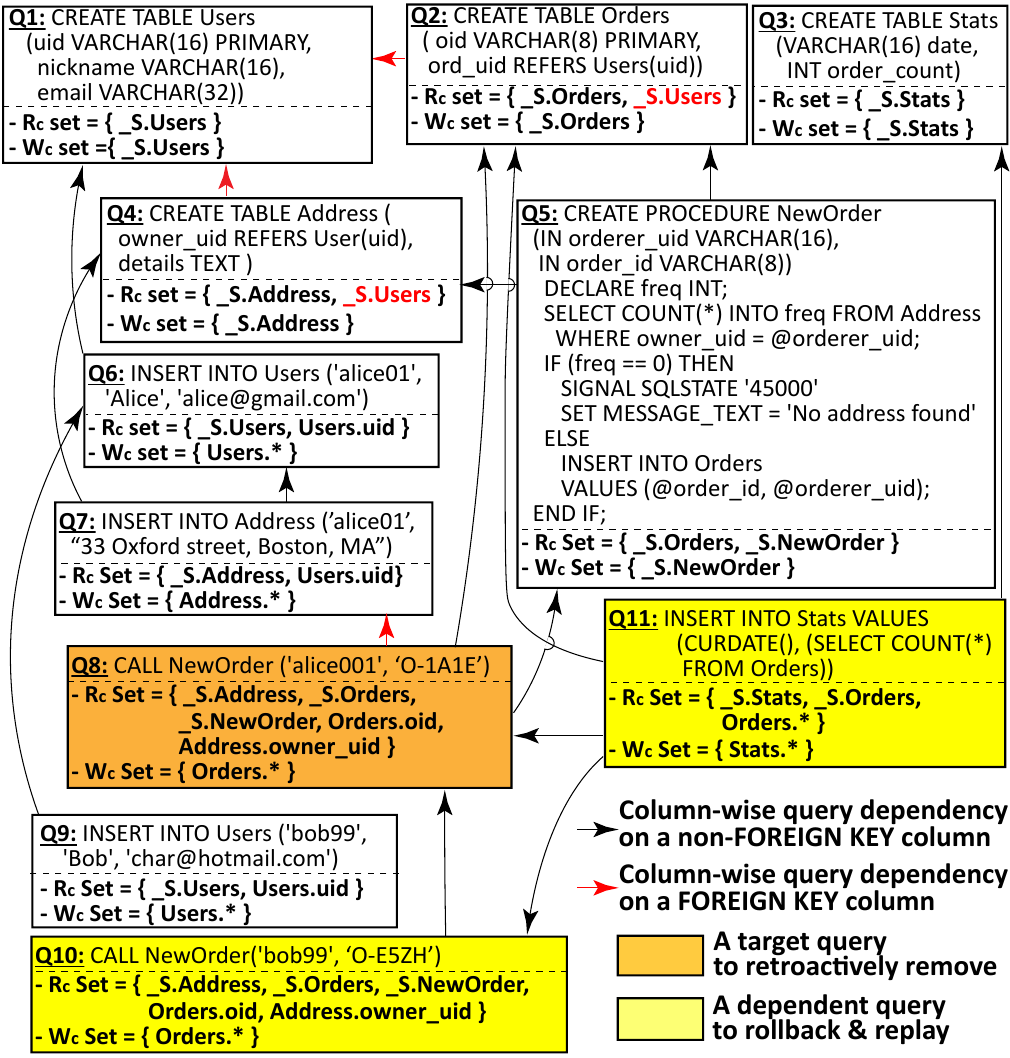}
  \caption{An E-commerce service's query dependency.}
  \label{fig:dependency-graph-website-columnwise}
\end{figure}

\subsection{Column-wise Dependency Analysis}
\label{subsec:columnwise-dependency-analysis}


\autoref{tab:set-policy} in \autoref{appendix:set-policy} presents all SQL statement types supported by \ultraverse for column-wise query dependency analysis. Each query is assigned a column-wise read set ($R_c$) and write set ($W_c$), based on which \ultraverse derives query dependencies. A query's $R_c$ set lists the columns of tables/views being read during its execution. The $W_c$ set lists the columns it writes. Besides the descriptions in \autoref{tab:set-policy}, we add the following remarks: \looseness=-1

\begin{itemize}

\item \textbf{Branch Conditions:} Control flow directions (e.g., \texttt{IF}, \texttt{CASE}) are challenging to  statically predict, since they often depend on the database's dynamically evolving state. \ultraverse handles this uncertainty by considering both directions of conditional branches found in \texttt{PROCEDURE}, \texttt{TRANSACTION}, or \texttt{CREATE TRIGGER} queries, and merging the $R_c/W_c$ sets from both true and false blocks. This approach overestimates the $R_c/W_c$ sets, potentially enlarging the dependency graph, but at this cost, we ensure the correctness of retroactive operations. This is because replaying more queries than minimally needed do not degrade the correctness. \looseness=-1

\item \textbf{Tracking Data Flows: } We track the data flows between SQL's \texttt{DECLARE} variables or (nested) queries in \texttt{PROCEDURE},  \texttt{TRANSACTIONs}, and \texttt{TRIGGER}. This is done by accumulating the $R$ set of all intermediate \texttt{SELECT} statements that mediate such data flows. Some examples are \texttt{"}\texttt{SELECT} \texttt{col\textsubscript{1}} \texttt{INTO} \texttt{var\textsubscript{a}} \texttt{FROM} \texttt{Table\textsubscript{1}} \texttt{WHERE..."}, or \texttt{"}\texttt{UPDATE} \texttt{Table\textsubscript{1}} \texttt{SET} \texttt{col\textsubscript{1}=} \texttt{(SELECT ...)}\texttt{ WHERE ..."}. \looseness=-1

\item \textbf{SQL Keywords Used in \texttt{PROCEDURE}:} Control flow or variable manipulation keywords used exclusively in \texttt{PROCEDUREs} such as \texttt{ITERATE}, \texttt{LEAVE LOOP}, \texttt{WHILE}, \texttt{REPEAT}, \texttt{DECLARE}, \texttt{SET variable}, \texttt{LABEL}, \texttt{CURSOR}, or \texttt{SIGNAL SQLSTATE} do not have any $R_c/W_c$ sets themselves. \looseness=-1

\item \textbf{A Query Accessing Multiple Databases: } The same $R/W$ set analysis policy in \autoref{tab:set-policy} applies to such a query, where each table name is prepended by its associated database name.  \looseness=-1

\item \textbf{Updatable \texttt{VIEWs}:} If a query performs \texttt{INSERT}, \texttt{UPDATE} or \texttt{DELETE} on an \texttt{VIEW}, the columns of the parent table/views that this \texttt{VIEW} references also get included in the query's $W_c$ set. \looseness=-1

\item \textbf{\texttt{AUTO\_INCREMENT}}: If a table's primary key column is set to \texttt{AUTO\_INCREMENT}, then the \texttt{INSERT} queries on this table include this primary key column in their $R_c$ set, as each insertion implicitly reads the latest primary key record. \looseness=-1

\item \textbf{\texttt{\_S.tablename:}} \texttt{\_S} is \ultraverse's virtual table for monitoring any changes in the database schema. Each column in \texttt{\_S} represents a schema object (i.e., table, view, procedure, or trigger). \texttt{\_S} is not an actual table in the database, but \ultraverse's internally reserved system variable. \ultraverse adds the \texttt{\_S.tablename} keyword in the $R/W$ sets of queries to remember their read/write access to the schema of the \texttt{"tablename"} table. \looseness=-1 

\end{itemize}

\begin{table}[t!]
\small
\setlength\tabcolsep{2.5pt}
\begin{tabular}{|l|}
\hline
\textbf{\underline{Notations}}\\
\begin{tabular}{ll}
$\bm{Q_n}$&: \textit{n-th} committed query\\
$\bm{\tau}$&: The retroactive target query's index\\
$\bm{R_c(Q_n)}$ | $\bm{W_c(Q_n)}$&: $Q_n$'s column-wise read/write set\\
$\bm{c}$&: a table's column\\
$\bm{Q_n} \rightarrow \bm{Q_m}$&: $\bm{Q_n}$ depends on $\bm{Q_m}$ \\
$\bm{A} \Longrightarrow \bm{B}$&: If $\bm{A}$ is true, then $\bm{B}$ is true\\
\end{tabular}\\
\textbf{\underline{Column-wise Query Dependency Rule}}\\
\setlength\tabcolsep{3.5pt}
\begin{tabular}{@{\hspace{-0.1cm}}ll@{\hspace{-0.0cm}}}
1.&$\exists c : ((c \in W_c(Q_m)) \wedge (c \in R_c(Q_n))) \wedge (m < n)$\\
&$  \Longrightarrow Q_n \rightarrow Q_m$\\
2.&$(Q_n \rightarrow Q_m) \wedge (Q_m \rightarrow Q_l) \Longrightarrow Q_n \rightarrow Q_l$\\
\end{tabular}\\
\hline
\end{tabular}
\caption{\ultraverse's column-wise query dependency rules.}
\label{tab:dependency-rule}
\end{table}

\ultraverse's query analyzer asynchronously derives the $R_c/W_c$ sets of all committed queries during regular service operations. The $R_c/W_c$ sets are used for efficient retroactive operations of adding, removing, or modifying past queries and updating the database accordingly. This process involves: \textit{(i)} undoing only those tables affected by the target query and its dependent queries; \textit{(ii)} modifying, adding, or removing the target query; \textit{(iii)} replaying only those queries dependent on the target query. \ultraverse uses a query dependency graph (\autoref{fig:dependency-graph-website-columnwise}) where nodes are queries and arrows are dependencies.\looseness=-1

In \ultraverse, one query is said to \textit{depend on} another if its execution could alter the latter's result. \autoref{tab:dependency-rule} defines \ultraverse's four column-wise query dependency rules.
Rule 1 states that if $Q_m$ writes to a column of a table/view and later $Q_n$ reads the same column, then $Q_n$ depends on $Q_m$. 
For instance, in \autoref{fig:dependency-graph-website-columnwise}, Q10$\rightarrow$Q8, because Q10 reads the \texttt{Orders.oid} column written to by Q8.
Note that our query dependency differs from the dependency in conflict graphs~\cite{trans}, which includes \textit{read-then-write}, \textit{write-then-read}, and \textit{write-then-write} cases. In contrast, our rule excludes the \textit{read-then-write} and \textit{write-then-write} cases, because a prior query's read/write operation on a data record does not affect the results of the later query that (over-)writes to the same record. 
Rule 2 states that if $Q_n$ depends on $Q_m$ and $Q_m$ depends $O_l$, then $Q_n$ also depends on $Q_l$ {(transitivity)}. For example, in \autoref{fig:dependency-graph-website-columnwise}, Q10$\rightarrow$Q7, because Q10$\rightarrow$Q8 (on \texttt{Orders.oid}) and Q8$\rightarrow$Q7 (on \texttt{Addresses.owner\_uid}). \looseness=-1 

Applying these read/write dependency rules to \texttt{\_S} (i.e., \ultraverse's virtual table for monitoring changes in the DB schema), we can derive the dependency rules between DDL queries and DML queries as follows: the DML queries that read a schema (e.g., \texttt{TABLE}/\texttt{VIEW} access, \texttt{PROCEDURE} call, or \texttt{TRIGGER}-linked query) depend on the DDL query committed earlier in time that wrote/modified this schema (i.e., \texttt{CREATE}, \texttt{DROP}, \texttt{ALTER}, or \texttt{RENAME}). Therefore, if a DDL query that modifies a schema is retroactively replayed, all dependent DDL/DML queries that post-access this schema are also to be replayed. In the example of \autoref{fig:dependency-graph-website-columnwise}, Q8 and Q10 that \texttt{CALL} the \texttt{NewOrder} procedure depends on Q5 that \texttt{CREATE}s it, thus if Q5 is the retroactive target query, then Q8 and Q10 should be also rolled back and replayed. \looseness=-1


In \autoref{fig:dependency-graph-website-columnwise}'s query dependency graph, we intentionally omit all queries whose $W_c$ set is empty (e.g., standalone \texttt{SELECT} queries), because such read-only queries do not affect the database's state during the \textit{what-if} replay. However, if a \texttt{SELECT} query is used inside \texttt{PROCEDURE}, \texttt{TRANSACTION}, \texttt{TRIGGER}, or inside another query as a nested query, then we merge the $R_c$ set of such \texttt{SELECT} queries into its wrapping query, because in such a case the \texttt{SELECT} query serves as an intermediate data flow between queries, tables, or SQL variables. As an additional remark, the red arrows in \autoref{fig:dependency-graph-website-columnwise} represent column-wise dependencies induced by \texttt{FOREIGN KEY} relationships-- if a column's value is retroactively changed, the \texttt{FOREIGN} \texttt{KEY} columns in other tables that reference this column may also be affected. In that sense,
adding the red arrows ensure rollback and replay of queries accessing such potentially affected foreign key columns. \looseness=-1

\autoref{appendix:columnwise-dependency-analysis}\footref{footnote:appendix} provides a proof of the column-wise dependency analysis.  \looseness=-1

\subsection{Row-wise Query Dependency Analysis}
\label{subsec:rowwise-analysis}
This section elaborates on \ultraverse's row-wise (horizontal) dependency analysis, an extension of the column-wise (vertical) dependency analysis described in \autoref{subsec:columnwise-dependency-analysis}. Essentially, if two queries access different table rows, their operations do not affect each other, and therefore their operations are mutually independent. \looseness=-1

\begin{table}[h!]
\setlength\tabcolsep{1pt}
\footnotesize
\begin{tabular}{|l|}
\hline
\texttt{Q6.INSERT INTO Users VALUES (`alice01',`Alice',`al@gmail.com')}\\
\texttt{\textcolor{white}{....}--> R{\tiny{row}}\footnotesize=\{\}, W{\tiny{row}}\footnotesize=\{ <Users.uid : `alice01'> \}}\\
\texttt{...}\\
\texttt{Q9.INSERT INTO Users VALUES (`bob99',`Bob',`bob@yahoo.com')}\\
\texttt{\textcolor{white}{....}--> R{\tiny{row}}\footnotesize=\{\}, W{\tiny{row}}\footnotesize=\{ <Users.uid : `bob99'> \}}\\
\texttt{...}\\
\texttt{Q12.UPDATE Users SET email=`alice@aol.com' WHERE uid=`alice01'}\\
\texttt{\textcolor{white}{....}--> R\tiny{row}\footnotesize = W{\tiny{row}}\footnotesize = \{ <Users.uid : `alice01'> \}}\\
\texttt{Q13.UPDATE Users SET email=`bob@hotmail.com' WHERE uid=`bob99'}\\
\texttt{\textcolor{white}{....}--> R\tiny{row}\footnotesize = W{\tiny{row}}\footnotesize = \{ <Users.uid : `bob99'> \}}\\
\hline
\end{tabular}
\caption{An extended scenario of \autoref{fig:dependency-graph-website-columnwise}.}. 
\label{tab:rowwise-queries}
\end{table}

\autoref{tab:rowwise-queries} provides an illustrative example, expanding \autoref{fig:dependency-graph-website-columnwise} by including Q12 and Q13. Q12 has a row-wise dependency on Q6, reading the same \texttt{Users} table's row that Q6 previously wrote (the row whose \texttt{uid} value is \texttt{`alice01'}). Similarly, Q13 depends row-wise on Q9, demonstrating a \textit{write-then-read} dependency on the row with \texttt{uid} value \texttt{`bob99'} in the \texttt{Users} table. However, Q13 is row-wise independent from Q6, because they access different rows in the \texttt{Users} table. Similarly, Q12 is row-wise independent from Q9, because they access mutually disjoint rows. Consequently, if Q6 is retroactively removed, (row-wise dependent) Q12 must be rolled back and replayed, but not (row-wise independent) Q9 or Q13. \looseness=-1

The crux of row-wise query dependency analysis is to pinpoint the table rows a query reads/writes based on the information in the query statement. In DML queries, \ultraverse chooses a specific column in each table as \textit{row identifier} column (or RI column), whose values represent the row IDs. In \autoref{tab:rowwise-queries}, \texttt{Users.uid} can serve as the \texttt{Users} table's RI column, because \texttt{Users.uid} is a primary key (i.e., each \texttt{uid} value is unique for each row). \looseness=-1

In column-wise dependency analysis (\autoref{subsec:columnwise-dependency-analysis}), $R_c$/$W_c$ elements represent the table columns a query reads/writes. In row-wise analysis, $R_r$/$W_r$ elements are denoted as \texttt{<RI\_columnName : value>}, indicating the RI column name and its value accessed by the query. In a \texttt{TRANSACTION}/\texttt{PROCEDURE}, the RI value of each executed query is either a constant or a symbolic expression found during DSE's execution path exploration (\autoref{subsec:transpilation}).
\autoref{tab:rowwise-queries} describes the $R_r$/$W_r$ sets of Q6, Q9, Q12, and Q13, which are DML queries (e.g., \texttt{INSERT, UPDATE, DELETE}). For DML queries like \texttt{UPDATE, DELETE} and \texttt{SELECT}, if their \texttt{WHERE} clause does not condition on any RI column's value of the table they access (say they access \texttt{table1}), then their $R_r$ is \texttt{{table1.colA : *}}, where \texttt{colA} is \texttt{table1}'s RI column and \texttt{*} is a wildcard. This wildcard implies that the query can potentially read any rows of \texttt{table1}. For DDL queries like \texttt{CREATE, ALTER} or \texttt{DROP}, their $R_r$ and $W_r$ sets include \texttt{".\_S.tablename = *"}, where the wildcard indicates the event that the schema of \texttt{
tablename"} has been changed. All DML queries implicitly include \texttt{"\_S.tablename = *"} in their $R_r$ set for the tables they read/write, so that if the table's schema ever retroactively changes, then all subsequent DML \& DDL queries operating on this table are regarded as dependent due to their \textit{write-then-read} relationship on \texttt{"\_S.tablename = *"} and thus they consequently roll back \& replay.  \looseness=-1

\ultraverse's row-wise dependency rule is formally detailed in \autoref{tab:dependency-rule2}, echoing \autoref{tab:dependency-rule}'s column-wise rule but based on $R_r$/$W_w$ sets instead of $R_c$/$W_c$ sets. Two queries are row-wise dependent if the earlier query's $W_r$ set and the later query's $R_r$ set has at least one common element whose RI column name and value both match. \looseness=-1

\para{Selection of an RI Column:} To execute row-wise analysis, \ultraverse must define at least one RI column per table in the target database. \ultraverse scans the query log and automatically selects suitable RI columns to maximize parallel query execution during retroactive replay (refer to \autoref{subsec:rollback-replay}). This approach aims to reduce the average runtime of retroactive operations. \looseness=-1

\begin{table}[t!]
\small
\setlength\tabcolsep{2.5pt}
\begin{tabular}{|l|}
\hline
\textbf{\underline{Notations}}\\
\begin{tabular}{ll}
$\bm{R_r(Q_n)}$ | $\bm{W_r(Q_n)}$&: $Q_n$'s row-wise read/write sets\\
$\bm{\langle c: v \rangle}$&: The element of $R_r$ or $W_r$ ($c$: column, $v$: value)\\
& \textcolor{white}{,} , indicating the table row(s) a query accesses\\
\end{tabular}\\
\textbf{\underline{Row-wise Query Dependency Rule}}\\
\setlength\tabcolsep{3.5pt}
\begin{tabular}{@{\hspace{-0.1cm}}ll@{\hspace{-0.0cm}}}
1.&$\exists c : ((\langle c: v \rangle \in W_r(Q_m)) \wedge (\langle c: v \rangle \in R_r(Q_n))) \wedge (m < n)$\\
&$  \Longrightarrow Q_n \rightarrow Q_m$\\
\end{tabular}\\
\begin{tabular}{l@{\hspace{-0.0cm}}}
* Rule 2 is the same as that in \autoref{tab:dependency-rule}.\\
\end{tabular}\\
\hline
\end{tabular}
\caption{\ultraverse's row-wise query dependency rules.}
\label{tab:dependency-rule2}
\end{table}

\para{Merging RI values:} Instances may arise when an \texttt{INSERT} query sets a new row's RI column \texttt{columnA}'s value as \texttt{value1}, and later an \texttt{"UPDATE SET columnA=..."} query alters this RI value to \texttt{value2}. In such a case, \texttt{<columnA : value1>} and \texttt{<columnA : value2>} potentially refer to the same physical row of the table. To handle such an ambiguity, \ultraverse treats \texttt{<columnA : value1>} and \texttt{<columnb : value2>} as identical, as \textit{merged} RI values. This conservative strategy increases the granularity of row-wise dependency separation, but ensures the accuracy of retroactive rollback and replay. \looseness=-1

\para{Alias RI Column}: \ultraverse can optionally select another column within the same table as the RI column's alias, enhancing row-wise dependency analysis's flexibility and performance. For example, in \autoref{tab:rowwise-queries}, consider we have a new query Q14: \texttt{"DELETE Users WHERE nickname=`Bob'"}. This query's \texttt{WHERE} clause does not specify the RI column value (i.e., \texttt{Users.uid}), but we can instead use \texttt{Users.nickname} as an \textit{alias} to \texttt{Users.uid}. Establishing \texttt{users.nickname} as an alias RI column allows \ultraverse to create a mapping from \texttt{<Users.nickname : 'Bob'>} to \texttt{<Users.uid : 'bob99'>}. \ultraverse's query analyzer learns this specific mapping when committing Q9 where Bob's data gets initially inserted into the \texttt{Users} table. Using this mapping, \ultraverse concludes that Q14's \texttt{R{\scriptsize{r}} = W{\scriptsize{r}} = {<Users.uid : `bob99'>}}. Note that if two RI values merge, then their associated alias mappings also merge.  \looseness=-1

\para{Operators in RI Value Expressions}: 
\ultraverse handles logical operators on RI column values and interprets them accordingly. For instance, in a \texttt{WHERE} clause, if an \texttt{AND} operator combines multiple column values, these columns formulate a multi-dimensional RI value, like \texttt{<tableName.(col1, col2) : (col1's value, col2's value)>}, similar to an SQL table's primary key comprising multiple columns. Queries access only those rows that meet both operands in the \texttt{AND} expression of the \texttt{WHERE} clause. Row-wise query dependency analysis captures this by computing the $R_w/W_r$ set intersection: two multi-dimensional RI values match only if their every degree's sub-value (e.g., RI col1's values, RI col2's values, ...) matches each other. Conversely, when a \texttt{WHERE} clause combines multiple column values with \texttt{OR} operators, the final $R_w/W_r$ set is a union of them, because such queries can potentially access any rows satisfying any operand in the \texttt{OR} expression. \ultraverse also handles arithmetic operators on RI column values: in such a case, the query's $R_r/W_r$ sets are regarded as symbolic expressions, and such a query's $R_r/W_r$ sets get concretized at the moment of retroactive operation, by resolving the symbolic values according to the retroactive scenario. 
\looseness=-1

\begin{wrapfigure}{r}{0.20\textwidth}
  \setlength{\columnsep}{0pt}
  \begin{center}
  \vspace{-10pt}
  \includegraphics[width=0.20\textwidth]{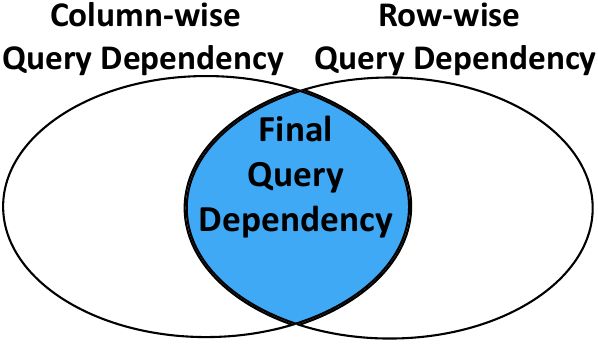}
  \vspace{-5pt}
  \end{center}
\end{wrapfigure}

Finally, \ultraverse regards two queries to have a dependency if and only if they are dependent both column-wise and row-wise (i.e., cell-wise), as depicted in the Venn diagram. 
Consequently, \autoref{subsec:columnwise-dependency-analysis}'s column-wise query dependency graph can be further reduced by applying the row-wise analysis. We show the full row-wise dependency graph for \autoref{fig:dependency-graph-website-columnwise} in \autoref{appendix:rowwise-dependency-graph}\footref{footnote:appendix}.  \looseness=-1

\autoref{appendix:rowwise-dependency-analysis}\footref{footnote:appendix} provides a proof of the row-wise dependency analysis.
\looseness=-1


\subsection{Efficient Rollback and Replay}
\label{subsec:rollback-replay}
Given a query dependency graph, \ultraverse rollbacks and replays only the queries dependent on the target query as follows:\looseness=-1

\begin{enumerate}
    \item \textbf{Rollback Phase}: This step rolls back each table that appears in any read or write set of the query dependency graph drawn for the target query, and then copies these tables into a temporary database. \looseness=-1
    \item \textbf{Replay Phase}: At this stage, the target query gets either added, removed, or modified as per the user's request. All the queries dependent on the target query are then replayed (within the temporary database) in parallel as much as possible, while ensuring the correctness of the final database state (i.e., ensuring strongly serialized commits). \looseness=-1
    \item \textbf{Database Update}: The original database gets locked, and the changes from the temporary database are reflected on the \textit{mutated} tables (to be explained) in the original database. Upon completion, the original database gets unlocked, and the temporary database gets deleted. \looseness=-1
\end{enumerate}

During the above process, every table in the original database is categorized as either a mutated, consulted, or irrelevant table. A table is considered \textit{mutated} if the table is included in the write set of at least one query dependent on the target query. \textit{Consulted} tables appear in the read set of dependent queries but not in the write set. Irrelevant tables are neither \textit{mutated} nor \textit{consulted} ones. \looseness=-1

During the rollback phase, \ultraverse uses either system versioning of temporal databases~\cite{system-versioning} or check-pointed backup databases to rollback \textit{mutated} and \textit{consulted} tables, as well as any logical \texttt{INDEX}es to their state at the first-accessed commit time after the retroactive operation's target time. \textit{Consulted} tables should be also rolled back because their former states are needed for access while replaying the dependent queries that update \textit{mutated} table(s).
Affected by this, other non-dependent queries that have those \textit{consulted} tables in their write set will also be replayed during the replay phase. During replay, the intermediate values of a \textit{consulted} table will be read by replayed queries; at the end of replay, \textit{consulted} tables will have the same state as before the rollback. \looseness=-1

In step 2's replay phase, the past commit order of dependent queries should be preserved, because otherwise, they could lead to inconsistency of the final database state-- leading to a different universe than the desired state. To ensure this  \textit{strict serializability} while achieving a fast replay speed, \ultraverse uses query dependency information to simultaneously execute multiple queries in parallel whose $R/W$ sets do not overlap. Such a parallel query execution is safe because if two queries access different table cells, as they do not cause a race condition to each other. This improves the replay speed while guaranteeing the same final database state as strongly serialized commits. A replay arrow from $Q_n \rightarrow Q_m$ is created if $n < m$ and the two queries have a conflicting operation~\cite{conflict} (a \textit{read-write}, \textit{write-read}, or \textit{write-write}) on the table/view's same column \textit{and} same RI value. \looseness=-1

Throughout the retroactive operation, \ultraverse continues to serve regular SQL operations to its clients. Retroactive operations are conducted on a temporary database, so normal service remains uninterrupted. After the rollback and replay phases, \ultraverse locks the original database, updates \textit{mutated} table tuples from the temporary database, and then unlocks it. Any new regular queries committed during the rollback and replay phases are accounted for in a new subsequent set of rollback and replay phases before unlocking the database. \looseness=-1

\para{Replaying Non-determinism:} During regular operations, \ultraverse records the concrete return value of each non-deterministic SQL function call (e.g., \texttt{CURTIME()} or \texttt{RAND()}). During the replay phase, \ultraverse enforces each query's non-deterministic function call to return the same value as before. In cases where a retroactively newly added query calls a timing function like \texttt{CURTIME()} that was not called during the regular operation, \ultraverse estimates its return value based on the (past) commit timestamps of its neighboring queries. Similarly, when replaying an \texttt{INSERT} query that had a primary key value determined by \texttt{AUTO\_INCREMENT}, the replay uses the same primary key value as in the past. If a new \texttt{INSERT} query is retroactively added, \ultraverse assigns a new primary key value that was never used in the past. \looseness=-1



\subsection{Hash-Jumper}
\label{subsec:hashjump}
During a retroactive operation, if the operation is expected not to change the final state of the database, the operation can be terminated in advance. This avoids unnecessary computational effort in situations where the retroactive operation doesn't change the final result. \looseness=-1

The \autoref{fig:hash-jumper} scenario illustrates the value of the Hash-jumper technique. Here, the \texttt{Membership} table (Q14) is created to record each user's membership level by using the \texttt{UpdateMembership} procedure (Q15). Alice's initial membership is initialized as `\texttt{gold}' (Q16), Bob's as `\texttt{gold}', and many more user memberships are initialized or updated in the subsequent operations.
Later at some point, suppoes that Alice's membership gets changed to `\texttt{diamond}' (Q99) due to her active purchases. Now, suppose we run a \textit{what-if} analysis to explore the effects of retroactively removing Alice's membership initialization (Q16).
During this retroactive operation, upon rolling back and replaying Q99, notice that the \texttt{Membership} table's state becomes the same as the query commit time of Q99 before the retroactive operation, because Alice's membership level gets overwritten to `\texttt{diamand}' upon executing Q99, regardless of what its value was before. Also, notice that all subsequent queries after Q99 are identical to the ones before the retroactive operation. Consequently, upon retroactively replaying Q99, we come to the realization that the final state of the \texttt{Membership} table after the whole retroactive operation will become the same as before the retroactive operation, and thus we can early-terminate the effectless replay operations (skipping the remaining replay of Q100$\sim$Q1000). \looseness=-1

\begin{figure}[t!]
  \includegraphics[width=.5\textwidth]{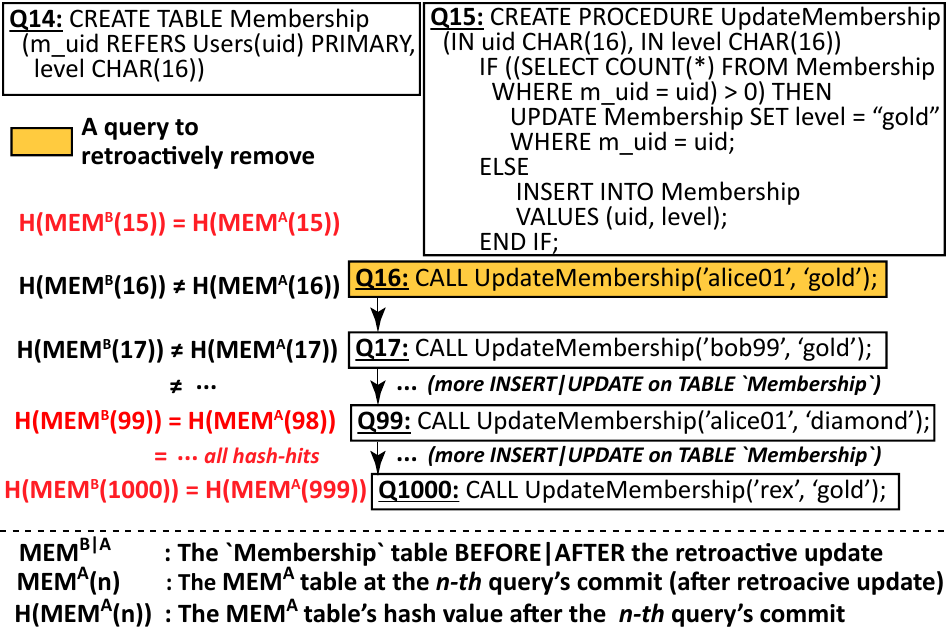}
  \caption{An example of Hash-jump.} 
  \label{fig:hash-jumper}
\end{figure}

\ultraverse introduces a mechanism called Hash-jumper to address such cases. The primary objective of Hash-jumper is to identify situations where replaying the remaining queries in the sequence will have no effect on the final state of the database. Upon recognizing such a situation, Hash-jumper terminates the retroactive operation.
In order to achieve this, \ultraverse logs the hash value of the modified tables at each query's commit during regular operations. During a retroactive operation's replay phase, Hash-jumper asynchronously runs in the background, comparing the hash value of the replayed query's output to the previously logged hash value. If the hash values match for all \textit{mutated} tables and there are no more queries to be retroactively added or removed after the current timestamp, a ``hash-hit" event is declared. This implies that the rest of the replay will simulate the same database evolution as before (taking all the same intermediate \texttt{IF} branch directions as before if there exist), and the final database state will be the same as before the retroactive operation. Thus, \ultraverse early-terminates the replay process and retains the original tables. \looseness=-1

The efficiency of the hash function is crucial in this scenario. To support large databases, an ideal hash function's computation time should not be influenced by the size of the table. To meet this requirement, \ultraverse designs an efficient table hash algorithm. The initial hash value of an empty table is set to 0. When a query is executed and the rows to be added or deleted are recorded, the query analyzer computes the hash value for each of these rows using a collision-resistant hash function (e.g., SHA-256). This hash value is then either added (for insertion) or subtracted (for deletion) from the target table's current hash value modulo p (i.e., the size of the collision-resistant hash function's output range). The time to compute the table hash for each query is constant with respect to the target table's size and is linear with the number of rows to be inserted or deleted. \looseness=-1

Given that the collision-resistant hash function's output is uniformly distributed, the collision rate of table hashes, irrespective of the number of rows in the tables, is bounded by $\frac{1}{p}$. In the case of SHA-256, this value is $2^{-256}$ ($< 10^{-77}$), which is negligibly smaller than the uncorrectable bit error rate (UBER) of today's flash or disk drives ($10^{-16} \sim 10^{-13}$). In that sense, Hash-jumper provides a robust and efficient mechanism for early termination of effectless retroactive operations.
See \autoref{appendix:hashing}\footref{footnote:appendix} 
for the proof and discussion on Hash-jumper's false positives \& negatives.
Nevertheless, \ultraverse offers the option of literal table comparison upon detecting hash-hits.
\looseness=-1

\section{Evaluation}
\label{sec:evaluation}

\begin{table}[t]
\setlength{\tabcolsep}{2.5pt}
\footnotesize
\renewcommand\arraystretch{1.2}
\centering
\begin{subtable}{0.5\textwidth}
\centering
\resizebox{\linewidth}{!}{
\begin{tabular}{|l||ccc|ccc|ccc|ccc|}
\toprule
\multirow{2}{*}{
{\textbf{Bench.}}}
  & \multicolumn{3}{c}{\textit{\textbf {250 Queries}}} &
\multicolumn{3}{c}{\textit{ \textbf {500 Queries}}} & \multicolumn{3}{c}{\textit{ \textbf {1K Queries}}} & \multicolumn{3}{c|}{\textit{ \textbf {2K Queries}}} \\ 
\cmidrule{2-4} \cmidrule{5-7} \cmidrule{8-10} \cmidrule{11-13}

 &\textsf{T+D}&\textsf{B}&\textsf{M}&\textsf{T+D}&\textsf{B}&\textsf{M}&\textsf{T+D}&\textsf{B}&\textsf{M}&\textsf{T+D}&\textsf{B}&\textsf{M} \\
\hline
Epinions &0.7s&0.7s&82.3s&0.7s&0.8s&710.2s&0.7s&1.3s&2.3H&0.7s&2.8s&>20.8H\\
TATP &0.6s&0.4s&34.5s&0.6s&0.8s&254.2s&0.6s&2.1s&2358s&0.8s&3.2s&8.3H\\
SEATS &1.4s&1.1s&$\times$&1.7s&1.9s&$\times$&1.9s&3.6s&$\times$&2.9s&8.4s&$\times$\\
TPC-C &1.3s&0.5s&131.2s&1.5s&0.8s&167.2s&1.6s&2.2s&1387s&2.4s&3.4s&4.3H\\
AStore &0.7s&0.6s&37.7s&0.7s&1.2s&291.2s&0.7s&2.5s&2824s&0.7s&6.3s&10.1H\\
\bottomrule
\end{tabular}
}
\caption{\textit{what-if} analysis time (``$\times$" means N/A).}\label{tab:comparison-speed}
\end{subtable}
    \hfil
\begin{subtable}{0.5\textwidth}
\centering
\resizebox{\linewidth}{!}{
\begin{tabular}{|ccc|ccc|ccc|ccc|}
\toprule

  \multicolumn{3}{|c}{\textit{\textbf {250 Queries}}} &
\multicolumn{3}{c}{\textit{ \textbf {500 Queries}}} & \multicolumn{3}{c}{\textit{ \textbf {1K Queries}}}  & \multicolumn{3}{c|}{\textit{ \textbf {2K Queries}}}  \\ 
\cmidrule{1-3} \cmidrule{4-6} \cmidrule{7-9} \cmidrule{10-12}

 \textsf{T+D}&\textsf{B}&\textsf{M}&\textsf{T+D}&\textsf{B}&\textsf{M}&\textsf{T+D}&\textsf{B}&\textsf{M}&\textsf{T+D}&\textsf{B}&\textsf{M} \\
\hline
43MB&60MB&3.5GB&43MB&62MB&10.9GB&48MB&62MB&40.8GB&54MB&64MB&>126GB\\
43MB&204MB&2.8GB&43MB&197MB&8.1GB&48MB&201MB&29.2GB&54MB&201MB&111GB\\
43MB&191MB&$\times$&56MB&197MB&$\times$&71MB&211MB&$\times$&92MB&244MB&$\times$\\
43MB&60MB&1.9GB&54MB&62MB&6.0GB&64MB&64MB&23.4GB&88MB&64MB&87.8GB\\
43MB&60MB&2.8GB&46MB&64MB&7.9GB&49MB&64MB&25.6GB&61MB&65MB&90.2GB\\
\bottomrule
\end{tabular}
}
\caption{RAM overhead for \textit{what-if} analysis.} 
\label{tab:comparison-overhead}
\end{subtable}

\caption{Speed of \ultraverse (\textsf{T+D}), Baseline (\textsf{B}) and Mahif~\cite{mahif} (\textsf{M})  across various transaction history sizes (250$\sim$2K queries).}
\label{tab:comparison2}
\end{table}

\begin{table}[t]
\setlength{\tabcolsep}{2.5pt}
{
\renewcommand\arraystretch{1.2}
\centering
\resizebox{0.5\textwidth}{!}{
\begin{tabular}{|l||cccc|cccc|cccc|}
\toprule
\multirow{2}{*}{
{{\textbf{Bench.}}}}
  & \multicolumn{4}{c}{\textit{\textbf {{\textsf{DB scale} = 1x}}}} &
\multicolumn{4}{c}{\textit{ \textbf {{\textsf{DB scale} = 10x}}}} & \multicolumn{4}{c|}{\textit{ \textbf {{\textsf{DB scale} = 100x}}}} \\ 
\cmidrule{2-5} \cmidrule{6-9} \cmidrule{10-13} 
 &\textbf{\textsf{DB\textsubscript{size}}}&\textsf{T+D}&\textsf{B}&\textsf{M}&\textbf{\textsf{DB\textsubscript{size}}}&\textsf{T+D}&\textsf{B}&\textsf{M}&\textbf{\textsf{DB\textsubscript{size}}}&\textsf{T+D}&\textsf{B}&\textsf{M}\\
\hline
{Epinions} &2.6MB&0.7s&0.4s&63.5s&20.1MB&0.8s&0.4s&63.4s&383MB&0.9s&0.7s&63.5s\\
{TATP} &38MB&0.6s&0.6s&60.1s&451MB&0.6s&0.6s&60.2s&7.8GB&0.7s&0.6s&60.4s\\
{SEATS} &50MB&1.4s&0.8s&$\times$&1.9GB&1.6s&1.0s&$\times$&4.7GB&1.7s&1.1s&$\times$\\
{TPC-C} &110MB&0.8s&0.8s&134.8s&1.1GB&0.9s&0.8s&135.7s&10.2GB&1.0s&0.8s&137.3s\\
{AStore} &60.5MB&0.7s&0.7s&36.6s&608MB&0.8s&0.8s&36.8s&5.99GB&0.9s&0.8s&37.6s\\
\bottomrule
\end{tabular}
}
}
\caption{\textit{what-if} analysis time across various DB sizes.}
\label{tab:comparison-db-size}
\end{table}

\para{Implementation:} \ultraverse's query dependency analysis (\autoref{sec:design}) is applicable to unmodified commodity DBMSes as far as they use SQL to manage databases and records committed queries (e.g., MySQL, Postgres, MariaDB). Further, \ultraverse's SQL transpiler (\autoref{sec:transpiler}) is compatible with general-purpose application programming languages that can be executed by dynamic symbolic execution (e.g., JavaScript, Python, C). Among these vast options, we chose MySQL and JavaScript as the \ultraverse prototype's host DBMS and host application language, as they are widely deployed.
We implemented the retroactive DBMS plugin (column-wise \& row-wise dependency analysis, replay scheduler, and Hash-jumper) in C++. The query analyzer reads MySQL's binary log~\cite{mysql-binary-log} to retrieve committed queries and computes each query's $R/W$ sets and table hashes. The replay scheduler uses a lockless queue~\cite{lockless-queue} and atomic compare-and-swap instructions ~\cite{compare-and-swap} to reduce contention among threads simultaneously dequeuing the queries to replay. For database rollback, there are 3 options: \textit{(i)} sequentially apply an inverse operation to every committed query; \textit{(ii)} use a temporal database to stage all historical table states; \textit{(iii)} assume periodic snapshots of backup DBs (e.g., every 3 days, 1 week). Our evaluation chose option 3 as it is the easiest approach in practice: creating system backups is a common practice and we can load a particular backup DB of our interest without incurring the rollback delay.
For the SQL transpiler, we implemented the JavaScript instrumentation software based on Jalangi2~\cite{jalangi2}, and used ExpoSE~\cite{Loring:2017} for path constraint resolution. We implemented the DSE analysis software of the server-side \& client-side JavaScript code based on the Chromium browser, Puppeteer API~\cite{puppeteer}, and MITMProxy~\cite{mitmproxy}. We used NodeJS and Python to implement the Z3-to-Python transpiler. \looseness=-1

\para{Test Setup:} We evaluated \ultraverse on Digital Ocean's VM equipped with 8 vCPUs, 16GB RAM, and 640GB SSD. Our evaluation compared three systems: (1) a baseline NodeJS application; (2) a transpiled NodeJS application with application-level transactions converted into compact SQL \texttt{PROCEDURE}s; (3) a non-transpiled NodeJS application with query dependency analysis using the retroactive DBMS plugin; and (4) a transpiled NodeJS application with query dependency analysis using the retroactive DBMS plugin. We denote these four systems as \textsf{B}, \textsf{T}, \textsf{D}, and \textsf{T+D}. \looseness=-1

Test cases were sourced from four BenchBase~\cite{oltp-bench} micro-benchmarks: TPC-C, TATP, Epinions, and SEATS. We converted the Java source code of these benchmarks JavaScript (NodeJS) for testing. The macro-benchmark used was AStore~\cite{astore}, an open-source e-commerce web application using the ExpressJS framework. AStore comprises 61 application-level transactions, of which 20 are database-updating transactions. We first compare the best one (\textsf{T+D}) with the baseline (\textsf{B}) and Mahif (\textsf{M})~\cite{mahif}(\autoref{subsec:comparison}). Then, we assess the performance (\autoref{subsec:performance}), overhead (\autoref{subsec:overhead}), and scalability (\autoref{subsec:scalability}) of the 4 systems (\textsf{B}, \textsf{T}, \textsf{D}, and \textsf{T+D}). Due to space limit, we provide \ultraverse's analysis details of all benchmarks (e.g., query dependency graphs) in \autoref{appendix:benchmarks}\footref{footnote:appendix}.
\looseness=-1

\subsection{Overall Comparison with Mahif}
\label{subsec:comparison}
We first compared \ultraverse's (\textsf{T+D}) effectiveness against the baseline (\textsf{B}) and Mahif (\textsf{M})~\cite{mahif}, a state-of-the-art \textit{what-if} DBMS. In particular, we used the open-source Mahif prototype~\cite{mahif-prototype}. Because Mahif demands a considerable amount of memory, we conducted the comparison on a bigger virtual machine with 12 CPU cores and 128GB RAM for Mahif. As Mahif is not scalable over the transaction history size and can only support a limited number of historical queries, we chose the history size of 250, 500, 1000, and 2000 queries, with the query dependency ratio of 50\%. Due to Mahif's inability to support string attributes, Mahif could not run the SEATS benchmark, since all the \texttt{UPDATE/DELETE/INSERT} queries in SEATS involved string attributes (marked as $\times$ in the tables). 
\looseness=-1

\para{\textit{What-if} Answering Speed: } \autoref{tab:comparison-speed} compares the \textit{what-if} analysis speed of \ultraverse and Mahif. In overall, \ultraverse (0.6s$\sim$2.9s) was 6450x faster than Mahif (34.5s$\sim$20.8H). The reason is because Mahif's runtime overhead increased exponentially as the number of historical queries increased, due to its growth of symbolic expression costs. On the other hand, \ultraverse only needed to parse and analyze the dependency of each query and replay only the dependent queries. Comparing \ultraverse with the baseline, \ultraverse's speed benefit by SQL transpilation and query dependency analysis became greater with the increasing number of queries. Aside the transaction history size, the database size did not significantly affect the \textit{what-if} analysis time of \ultraverse and Mahif (see \autoref{tab:comparison-db-size}), because their number of queries to be replayed was unaffected by the database size. \looseness=-1

\para{Overhead: } \autoref{tab:comparison-overhead} compares the memory usage of \ultraverse and Mahif. In overall, \ultraverse (43$\sim$92MB) was 1370x more memory-efficient than Mahif (requiring 1.9G$\sim$126GB). The memory required by \ultraverse is primarily used for temporarily storing each query statement, replaying each dependent query, and maintaining the dependent $R/W$ sets, while Mahif required a steeply growing amount of memory with the number of queries. \looseness=-1

\para{Correctness: } Unlike \ultraverse, Mahif failed to preserve the accuracy of application-level semantics during its \textit{what-if} simulation for TPC-C, SEATS, and TPCC transactions, because Mahif's analysis only handles individual SQL queries and does not capture the application-level data flows as a whole. Moreover, Mahif's \textit{what-if} analysis did not support the use of SQL \texttt{TRANSACTION}/\texttt{PROCEDURE} that integrate multiple queries. Mahif also lacked support for certain SQL features and query types, such as using string, boolean, or date/time type as column attributes, native SQL APIs (e.g., \texttt{CURRENT\_TIMESTAMP()} or \texttt{RANDOM()}), or DML (\texttt{CREATE, DROP, ALTER}) queries. Due to these limitations, Mahif didn't preserve the correctness of \textit{what-if} analysis in application-level or SQL-level transactions. \looseness=-1

\subsection{Detailed Performance}
\label{subsec:performance}

\para{\textit{What-if} Simulation Speed: } 
Next, we evaluated the performance of the \textsf{B}, \textsf{T}, \textsf{D}, and \textsf{T+D} versions. Note that unlike Mahif, all these 4 systems preserve correctness of the application-level \textit{what-if} simulation. We used a backup database as the starting point and a transaction history of 1 million queries. Each \textit{what-if} simulation test selected 1\% queries in the transaction history for retroactive operation. The transaction commit history was then replayed on the backup database based on these retroactive modifications.
\looseness=-1

\begin{table}
\renewcommand\arraystretch{1.2}
   \footnotesize
\centering
\begin{subfigure}[t]{0.5\textwidth}
\centering
{\includegraphics[width=\linewidth]{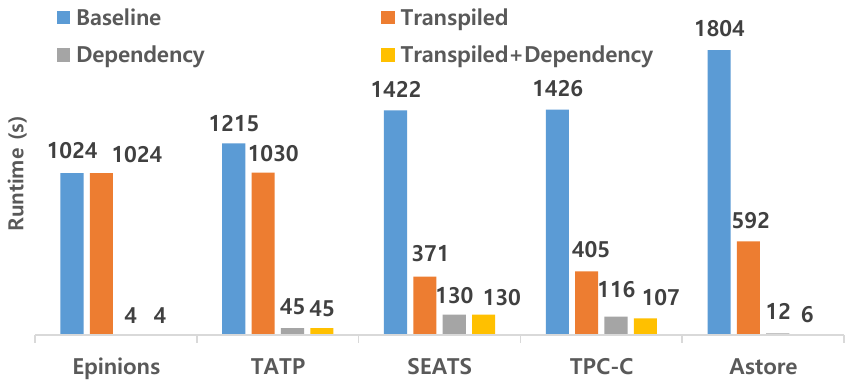}}
\caption{\textit{What-if} analysis runtime for 1 million queries.}
\label{fig:retroactive-operation}
\end{subfigure}

\begin{subtable}[t]{0.5\textwidth}
  \centering  
  \begin{tabular}{|c||c|c|c|c|c|c|}  
    \hline
    & Epinions & TATP & SEATS & TPC-C & AStore \\\hline\hline
    \textbf{at 10\%} & 5s & 52s & 131s & 107s & 8s \\\hline
    \textbf{at 25\%} & 12s & 128s & 65s & 51s & 20s \\\hline
    \textbf{at 50\%} & 23s & 257s & 0.18H & 0.14H & 40s \\\hline
    \textbf{at 100\%} & 46s & 512s & 0.36H & 0.29H & 77s \\\hline
  \end{tabular}
\caption{Hash-jumper's runtime across various hash-hit points.}  
  \label{tab:hash-jumper}
\end{subtable}

\begin{subtable}[t]{0.5\textwidth}
   \footnotesize
  \centering
  \begin{tabular}{|c||c|c|c|c|c|c|}  
    \hline
    & Epinions & TATP & SEATS & TPCC & AStore \\\hline\hline
    \textbf{\textsf{B}} & 1.11ms & 2.62ms & 7.99ms & 32.8ms & 9.1ms \\\hline 
    \textbf{\textsf{T/T+D}} & 1.11ms & 2.06ms & 6.79ms & 8.6ms & 7.1ms \\\hline    
  \end{tabular}
  \caption{Runtime of regular application-level transactions.}  
  \label{tab:regular-transaction-speed}
\end{subtable}

\caption{Performance of \ultraverse.}
\label{tab:performance}
\end{table}

Different replay strategies were adopted for the four systems. The baseline (\textsf{B}) version serially replayed all committed application function calls of the unmodified application-level transactions-- serial execution is necessary to ensure strong serializability of the \textit{what-if} replay. The transpiled (\textsf{T}) version serially replayed the function calls of the transpiled application-level transactions. The dependency-analyzed (\textsf{D}) version replayed the non-transpiled application-level transactions based on column-wise and row-wise query dependency analysis. The \ultraverse (\textsf{T+D}) version directly replayed the raw SQL \texttt{PROCEDURE}s from the log based on column-wise and row-wise query dependency analysis. \looseness=-1

\autoref{fig:retroactive-operation} compares the runtime of \textit{what-if} simulation for the transaction history window of 1 million queries.
The \textsf{T+D} version was 23.60x faster on average than the \textsf{B} version, and the \textsf{T} version was 2.01x faster on average than the \textsf{B} version. The \textsf{T} version gained speed primarily due to the SQL transpilation technique, which merged the calls of individual queries belonging to the same application-level transaction into a single SQL \texttt{PROCEDURE} call, thus reducing the number of round-trip communications between the DBMS client and server. This advantage became more pronounced when each application-level transaction contains a higher number of query executions by using loops (e.g., SEATS, TPC-C, AStore). The \textsf{D} version did not get such a speedup benefit of transpilation, but still gained speedup by replaying only dependent application-level transactions. \looseness=-1

The \textsf{D} version and \textsf{T+D} version gained a speed from column-wise and row-wise query dependency analysis, which reduced the number of queries to replay and increased the level of query execution parallelism during replay. Benchmarks such as Epinions, TATP, and AStore benefited from both replay query reduction and parallel replay. SEATS and TPC-C benefited only from replay parallelism, as all transactions in these two benchmarks depended on the retroactive operation's target transactions in both column-wise and row-wise terms. \looseness=-1

In all 4 system versions, reading the query log and replaying queries were performed in parallel. However, because the disk batch-reading speed of SSD NVMe was faster than the DBMS's serial execution of queries, the critical path of \textit{what-if} simulation was the DBMS's replay runtime rather than the disk I/O of query logs. \looseness=-1

\para{Hash-jumper: }
The \ultraverse version gained an additional speedup with its Hash-jumper feature enabled, as shown in \autoref{tab:hash-jumper}. This test, unlike the one referred to in \autoref{fig:retroactive-operation}, used 10 million queries. The Hash-jumper prototype has an upper-bound hash collision rate of approximately {$1.16 x 10^{-77}$}. Each benchmark's test cases were set up to have a hash-hit at various timestamp locations (10\%, 25\%, 50\%, and 100\%) in the transaction history. The sooner a hash-hit was detected by \ultraverse, the sooner the \textit{what-if} replay terminated, finishing the effectless operation. The test case with no hash-hit (i.e., at 100\%) implicitly measures the overhead of enabling Hash-jumper. This overhead was on average 2.4\%, reflected in the slowdown of the replay speed compared to disabling Hash-jumper. However, this overhead can be compensated by increasing the number of CPUs and RAMs, since the hash computation can be done asynchronously and independently from the replay.

\begin{table}

\begin{subtable}[t]{0.5\textwidth}
  \footnotesize
  \centering
\resizebox{\linewidth}{!}{
  \begin{tabular}{|l||c|c|c|c|c|}  
    \hline
    & Epinions & TATP & SEATS & TPC-C & AStore \\\hline\hline
    \textbf{Analysis Time} & 21.3s & 31.7s & 145.2s & 151.3s & 187.8s \\\hline
  \end{tabular}
  }
  \caption{SQL transpiler's application code analysis time.}
  \label{tab:transpilation} 
\end{subtable}

\begin{subtable}[t]{0.5\textwidth}
   \footnotesize
  \centering
\resizebox{\linewidth}{!}{
  \begin{tabular}{|l||c|c|c|c|c|c|}  
    \hline
    & Epinions & TATP & SEATS & TPC-C & AStore \\\hline\hline
    \textbf{MySQL's Query Log} & 291b & 376b & 524b & 489b & 441b \\\hline
    \textbf{\ultraverse's Log} & 12b & 48b & 35b & 110b & 76b \\\hline
  \end{tabular}
  }
  \caption{Average log size per query (bytes).}
  \label{tab:ultraverse-space-overhead} 
\end{subtable}

\begin{subtable}[t]{0.5\textwidth}
   \footnotesize
  \centering  
  \begin{tabular}{|c||c|c|c|c|c|c|}  
    \hline
    & TPC-C & TATP & Epinions & SEATS & RS \\\hline\hline
    \textbf{\textsf{T+D}} & 7.0\% & 0.7\% & 3.9\% & 3.9\% & 6.4\% \\\hline 
    \textbf{\textsf{T+D+H}} & 9.5\% & 0.7\% & 4.9\% & 4.2\% & 6.4\% \\\hline
  \end{tabular}
  \caption{Query dependency logger's overhead.}  
  \label{tab:logging-overhead}
\end{subtable}

\begin{subtable}[t]{0.5\textwidth}
   \footnotesize
  \centering
\begin{tabular}{|l||c|c|c|c|c|c|}  
    \hline
    & TPC-C & TATP & Epinions & SEATS & AStore \\\hline\hline
    \textbf{Slowdown} & 8.2\% & 7.2\% & 14.1\% & 16.5\% & 13.3\% \\\hline
  \end{tabular}
  \caption{Overhead of simultaneously running \textit{what-if} simulation and regular operations in the same machine.}
  \label{tab:simultaneous-overhead} 
\end{subtable}

\caption{Overhead of \ultraverse.}
\label{tab:overhead}
\end{table}

\para{Regular Transaction Speed: }
\ultraverse's SQL transpiler merges the queries executed separately within the same application-level transaction function into a single SQL \texttt{PROCEDURE}, effectively reducing the RTT between the DBMS client and server. This transpilation technique can be used not only for \textit{what-if} analysis, but also for speeding up the processing of regular application-level transactions. \autoref{tab:regular-transaction-speed} compares the average transaction runtime of each benchmark for the baseline and transpiled versions. On average, the baseline version took 10.7ms to process a transaction, whereas the transpiled version took 5.13ms-- 105.1\% faster. 
The speed of Epinions was the same, because its each transaction was a single query. On the other hand, SEATS, TPC-C, and AStore had a greater speedup because they had loops of executing multiple queries in a single transaction. 
In this local machine test, the RTT between the DBMS client and server was about 1ms. However, in a remote DBMS setup with a larger RTT (e.g., >20ms), the speedup from the transpiled version would be multiplied by the increase in RTT. Meanwhile, this speedup due to RTT gradually diminishes as each transaction requires more time for local computation (to be discussed in \autoref{subsec:scalability}).
\looseness=-1


\subsection{Overhead}
\label{subsec:overhead}

\para{Transpilation Time: }
The SQL transpilation time consists of application code instrumentation, path exploration with dynamic symbolic execution, and the final conversion into SQL code. 
Among the benchmarks, TPCC used \textit{for}-loops (e.g., batch-ordering products); Astore used black-box APIs (e.g., native JavaScript APIs unmodeled by Z3); and a majority of the benchmarks used SQLSTATE singals within various error-detecting \textit{if}-branches (e.g., no flight seats available).
As shown in \autoref{tab:transpilation}, the transpilation time for each benchmark varied from 21.3 seconds to 187.8 seconds. This time increased with the number of application-level transactions or the number of execution paths (e.g., loops). This transpilation overhead only occurs once per application during offline code analysis. \looseness=-1

\begin{table}[t]
\fontsize{6}{7}\selectfont
\centering
\setlength{\tabcolsep}{2.2pt}

\renewcommand\arraystretch{1.2}
\centering
\begin{subtable}{0.5\textwidth}
\centering
\begin{tabular}{|l||cccc|cccc|cccc|}
\toprule
\multirow{2}{*}{
{\textbf{Bench.}}}
  & \multicolumn{4}{c}{\textit{\textbf {History Size: 1M}}} &
\multicolumn{4}{c}{\textit{ \textbf {: 10M}}} & \multicolumn{4}{c|}{\textit{ \textbf {: 100M}}} \\ 
\cmidrule{2-5} \cmidrule{6-9} \cmidrule{10-13}
&\textsf{B}&\textsf{T}&\textsf{D}&\textsf{T+D}&\textsf{B}&\textsf{T}&\textsf{D}&\textsf{T+D}&\textsf{B}&\textsf{T}&\textsf{D}&\textsf{T+D} \\
\hline
Epinions &1024s&1024s&4s&4s&2.85H&2.85H&48s&45s&28.6H&28.6H&501s&490s\\
TATP &1215s&1030s&57s&45s&3.3H&1.1H&1277s&505s&32.9H&11.6H&4.1H&1.5H\\
SEATS &1422s&371s&465s&130s&3.95H&1.02H&1.4H&0.4H&39.4H&10.2H&13.9H&3.6H\\
TPC-C &1426s&405s&369s&116s&4.0H&1.1H&0.9H&0.3H&39.6H&18.4H&6.4H&2.9H\\
AStore &1804s&592s&31s&12s&5.0H&1.7H&418s&135s&50.4H&17.0H&1.0H&0.4H\\
\bottomrule
\end{tabular}
\caption{\textit{What-if} analysis time for various transaction history sizes.}
\label{tab:retroactive-windows}
\end{subtable}
    
\begin{subtable}{0.5\textwidth}
\centering
\begin{tabular}{|l||ccc|ccc|ccc|}
\toprule
\footnotesize
\multirow{2}{*}{
{\textbf{Bench.}}}
  & \multicolumn{3}{c}{\textit{\textbf {DB Size: 1x}}} &
\multicolumn{3}{c}{\textit{ \textbf {: 5x}}} & \multicolumn{3}{c|}{\textit{ \textbf {: 10x}}} \\ 
\cmidrule{2-4} \cmidrule{5-7} \cmidrule{8-10} &\textsf{T}&\textsf{D}&\textsf{T+D}&\textsf{T}&\textsf{D}&\textsf{T+D}&\textsf{T}&\textsf{D}&\textsf{T+D} \\
\hline
Epinions &1.0x&256x&256x&1.0x&256x&256x&1.0x&256x&256x\\
TATP &1.3x&21.3x&27.0x&1.3x&21.5x&26.2x&1.3x&21.5x&25.8x\\
SEATS &3.8x&3.1x&10.9x&3.7x&3.2x&10.9x&3.7x&3.2x&10.5x\\
TPC-C &3.2x&3.9x&13.4x&3.0x&3.8x&12.7x&2.7x&3.9x&11.5x\\
AStore &3.1x&99.4x&150x&3.0x&96.8x&148x&3.0x&97.6x&145x\\
\bottomrule
\end{tabular}
\caption{Speed against Baseline across various DB sizes.}
\label{tab:database-sizes} 
\end{subtable}
    
\begin{subtable}{0.5\textwidth}
\centering
\begin{tabular}{|l||ccc|ccc|ccc|ccc|}
\toprule
\footnotesize
\multirow{2}{*}{
{\textbf{Bench.}}}
  & \multicolumn{3}{c}{\textit{\textbf {Dependency: 1\%}}} &
\multicolumn{3}{c}{\textit{ \textbf {: 10\%}}} & \multicolumn{3}{c}{\textit{ \textbf {: 50\%}}} & \multicolumn{3}{c|}{\textit{ \textbf {: 100\%}}} \\ 
\cmidrule{2-4} \cmidrule{5-7} \cmidrule{8-10} \cmidrule{11-13} &\textsf{T}&\textsf{D}&\textsf{T+D}&\textsf{T}&\textsf{D}&\textsf{T+D}&\textsf{T}&\textsf{D}&\textsf{T+D}&\textsf{T}&\textsf{D}&\textsf{T+D} \\
\hline
Epinions &1.0x&366x&366x&1.0x&35.5x&35.5x&1.0x&6.7x&6.7x&1.0x&3.6x&3.6x\\
TATP &1.3x&343x&421x&1.3x&34.4x&41.7x&1.3x&6.0x&8.2x&1.3x&3.2x&4.1x\\
SEATS &-&-&-&-&-&-&-&-&-&3.8x&3.2x&11.2x\\
TPC-C &-&-&-&-&-&-&-&-&-&3.2x&4.4x&13.4x\\
AStore &3.1x&351x&836x&3.1x&34.7x&94.3x&3.1x&7.2x&18.3x&3.1x&3.4x&9.3x\\
\bottomrule
\end{tabular}
\caption{Speedup against Baseline across various query dependencies.}
\label{tab:dependency-rates} 
\end{subtable}
\caption{Scalability of \ultraverse.}
\label{tab:comparison}
\end{table}

\para{Query Dependency Log Size: }
\autoref{tab:ultraverse-space-overhead} compares the size of MySQL's binary query log and the additional log generated by \ultraverse. On average, MySQL's binary log is 424 bytes per query, while \ultraverse introduces an additional 56 bytes per query, resulting in a 7.6\% increase in the log storage overhead.
\looseness=-1

\para{Logging Read/Write Sets and Hashes: }
During regular operations, \ultraverse's DBMS plugin runs an asynchronous background process that derives column-wise and row-wise $R$/$W$ sets of committed queries and table hashes. This incurs additional computational overhead. The overhead this background logging activity imposes on the regular query processing operation of the DBMS is measured in \autoref{tab:logging-overhead}, where \textsf{T+D} refers to the \ultraverse version and \textsf{T+D+H} refers to \ultraverse with Hash-jumper enabled. Across all benchmarks, the overhead from the logging activity ranged from 0.6\% to 9.5\%. However, this overhead can be offset in a practical setting by copying the DBMS's MySQL binary log to another machine and having \ultraverse generate its $R$/$W$ set logs on that machine. This is feasible because \ultraverse's logging daemon operates asynchronously and independently from the DBMS's regular operations.
\looseness=-1

\para{Concurrent \textit{What-if} Simulation with Regular Operations: }
The impact on the DBMS's regular operations while the \textit{what-if} simulation is concurrently running is shown in \autoref{tab:simultaneous-overhead}. The average slowdown overhead ranged from 3.3\% to 16.5\%. This overhead can also be reimbursed if the \textit{what-if} simulation is run on a different machine, using a copied backup database and transaction history.
\looseness=-1

\subsection{Scalability}
\label{subsec:scalability}

\noindent\textbf{Transaction History Size: }
The runtime of the \textit{what-if} simulation was measured across different transaction history sizes: 1M, 10M, and 100M queries. \autoref{tab:retroactive-windows} compares the runtimes of the 4 versions (i.e., \textsf{B, T, D,} and \textsf{T+D}). The runtimes of all 4 systems generally increased proportionally with the size of the transaction history, as their number of transactions to replay also increased linearly.
\looseness=-1

\para{Database Size: }
The runtime of the \textit{what-if} simulation was also measured across different database sizes: 1x, 5x, and 10x,  ranging from 10MB to 10GB in overall. The speedup of the \textsf{T}, \textsf{D}, and \textsf{T+D} versions compared to the \textsf{B} version is shown in \autoref{tab:database-sizes}. The speedup remained generally consistent as the database size increased. However, benchmarks with computationally intensive transactions, such as SEATS and TPC-C, experienced more slowdowns for \textsf{T} and \textsf{T+D}. This occurred because the benefits of RTT reduction from transpilation were weakened by the increased computational time for transactions due to the increased database size (e.g., larger tables takes a more CPU time for JOIN operations). \looseness=-1

\para{Query Dependency Rate: }
The runtime of the \textit{what-if} simulation was measured across different query dependency rates: 1\%, 10\%, 50\%, and 100\%. The speed of \textsf{T} remained roughly the same regardless of changes in the query dependency rate since it had to replay all queries in the transaction history, anyway. For the \textsf{D} and \textsf{D+T} versions, as the query dependency rate increased and they had to replay more queries than before, the speedup rate decreased. However, even when the dependency rate was 100\% (requiring the replay of all queries), the \textsf{D} and \textsf{D+T} versions were still faster than the \textsf{B} and \textsf{T} versions. This is because query dependency analysis enabled spontaneous parallel replay of multiple queries with non-conflicting $R$/$W$ sets, resulting in a speedup ranging from 2.7x to 4.2x. For the SEATS and TPC-C benchmarks, we report only for a 100\% query dependency rate, because in those benchmarks almost all transactions depended on each other from the global view of the entire transaction history due their complex data record correlations. However, it's important to note that their multiple consecutive queries could still be spontaneously replayed in parallel if their column-wise or row-wise $R$/$W$ sets had no \textit{read-write}, \textit{write-write}, and \textit{write-read} conflicts. \looseness=-1



\section{Discussion}
\label{appendix:discussion}

\para{Using \ultraverse for Concurrency Control:}
Besides \textit{what-if} analysis, \ultraverse can be also used for efficient concurrency control to \textit{safely} execute multiple transactions in parallel. For example, Calvin~\cite{calvin} and Bohm~\cite{bohm} are transaction schedulers designed to efficiently manage distributed or versioned databases. However, their usability is limited because they cannot learn about a transaction's read/write sets without executing it. To handle this limitation, they use a read-lock detection scheme to identify transactions doing dirty reads, and the whole scheduling restarts whenever this happens. To eliminate such inefficiency of expensively restarting transactions, the schedulers can adopt \ultraverse's fine-grained query dependency analysis to attain the prior knowledge of transaction dependency and do accurate scheduling. \looseness=-1



\para{Managing Many \textit{what-if} Scenarios: } \ultraverse manages different versions of \textit{what-if} transaction history by attaching \textit{what-if} tags on each node of the transaction dependency graph to mark checkpoints as the starting point of a new \textit{what-if} scenario. The \textit{what-if} tags are essentially used to mark a branch that creates a different universe (i.e., different final database state). All such universes collectively form \ultraverse. \looseness=-1

\para{Replaying Interactive Human Decisions: } Besides \ultraverse's \textit{what-if} analysis based on a software's system-level (i.e., SQL-level and application-level) dependency, the analysis could be further enhanced by incorporating the realistic replay of interactive \textit{human minds}. 
For instance, in a \textit{what-if} analysis of a stock market, we can realistically replay the interactive decisions of the stock-trading users. The extended version of \ultraverse supports configurable human decisions, which is implemented based on trigger rules. For example, \ultraverse can be configured to suppress Alice's \texttt{StockPurchase} transaction during the \textit{what-if} replay if a certain stock symbol's price goes higher than Alice's buy-threshold. Such triggering rules can be either configured manually by a \textit{what-if} analyst or assisted by an instruction-tuned LLM~\cite{instruction-tuning} to simulate interactive human decisions upon replaying each user transaction. 
\looseness=-1

\para{Symbol Explosion: } There could be corner cases where a path explosion (i.e., circular-dependent loop conditions) occurs within dynamically generated code (e.g. \texttt{eval()}). Such execution paths within dynamic code are difficult to detect. This is fundamentally a challenging problem in the research field of DSE. In such circumstances,  \ultraverse's dynamic symbol spawning technique (\autoref{subsec:dse-discussion}) would spawn as many symbols as the number of paths executed within the dynamic code. We regard this an interesting topic for future research and aim to devise a more efficient solution.

\para{Cross-app \& Cross-network \textit{What-if} Analysis: }  \ultraverse considers remote network entities as a blackbox, creating a new \textit{network-data} symbol for each network API call. Our future work aims to enable precise \textit{what-if} analysis across multiple application services.
\looseness=-1

\section{Related Work}
\label{sec:related}


\noindent\textbf{Temporal Databases}~\cite{temporal-database-intro} store each record with a reference of valid time, denoting the beginning and ending time a database system regards a record to be valid in its database~\cite{temporal-valid-transaction}. Several database languages support temporal queries~\cite{temporal-survey, temporal-language}. SQL:2011~\cite{temporal-sql} incorporates temporal features and MariaDB~\cite{mariadb-code} supports its usage. 

\noindent \textbf{Database Versioning} (e.g., OrpheusDB~\cite{rdbms-versioning} or Apache Iceberg~\cite{iceberg}) allows users to efficiently store and load past versions of SQL tables by traversing a version graph, which is similar to how \ultraverse manages different \textit{what-if} universes (\autoref{subsec:dse-discussion}). While temporal and versioning databases can query a database's past states, they do not support \textit{efficient} retroactive replay of transactions like \ultraverse. \looseness=-1


\para{Database Recovery/Replication} efficiently recovers a database by logging values (SiloR~\cite{silor} Kuafu~\cite{kuafu}) or by replaying queries (~\cite{checkpoint-recovery, ganymed, lazy-replica, remusdb, nondeterminism}). But they are not designed for retroactive modification of past transactions with strong serialization~\cite{strong-serializability} like \ultraverse. 

\para{Attack Recovery:} Several methods have been developed to recover from attacks. CRIU-MR~\cite{criu-mr} recovers a malware-infected Linux container by selectively removing malware during checkpoint restoration. ChromePic~\cite{chromepic} replays browser attacks by transparently logging the user's page navigation activities and reconstructing the logs for forensics. RegexNet~\cite{regexnet} identifies malicious DoS attacks of sending expensive regular expression requests to the server, and recovers the server by isolating requests containing the identified attack signatures. However, the prior works do not address how to retroactively undo the damages on the server's persistent database, both efficiently and correctly from application semantics.
Warp~\cite{warp} and Rail~\cite{rail} selectively remove problematic user request(s) or patch the server's code and reconstruct the server's state accordingly. However, these techniques require replaying the heavy browsers ($\sim$100MB per instance) during their replay phase, which is not scalable for large services that have a large number of users or transactions (e.g., $\geq$ 10M) to replay. On the other hand, \ultraverse's strength lies in its scalable efficiency: \ultraverse transpiles the application code into the equivalent SQL \texttt{PROCEDURE} which contains only the minimally required SQL queries (affecting the database's state), whose \textit{what-if} replay is faster and lighter-weight. Further, \ultraverse proposes that novel DBMS-level optimization techniques: column-wise \& row-wise query dependency analysis and Hash-jumper. \looseness=-1

\para{Provenance in Databases:} 
\textit{What-if-provenance}~\cite{caravan} speculates the output of a query if a hypothetical modification is made to the database. \textit{Why-provenance}~\cite{why-provenance} traces the origin tuples of each output tuple by analyzing the lineage~\cite{lineage} of data generation. \textit{How-provenance}~\cite{how-provenance} explains the way origin tuples are combined to generate each output (e.g., ORCHESTRA~\cite{orchestra}, SPIDER~\cite{spider}).
\textit{Where-provenance}~\cite{why-provenance} traces each output's origin tuple and column, annotates each table cell, and propagates labels~\cite{annotation} (e.g., Polygen~\cite{polygen}, DBNotes~\cite{dbnotes}).
Mahif~\cite{mahif} is a recent work that answers historical \textit{what-if} queries, which computes the delta difference in the final database state given a modified past operation. Mahif leverages symbolic execution to ignore a subset of transaction history that is provably unaffected by the modified past operation. However, unlike \ultraverse, Mahif is not scalable over the transaction history size (\autoref{subsec:comparison}). 
Most importantly, all prior database provenance works do not capture application code's semantics during \textit{what-if} analysis.

\para{SQL transpiler:} 
QBS~\cite{qbs} is a system that uses the theory of finite ordered relation~\cite{ordered-relation} to selectively transform fragments of application logic into SQL queries to enhance an application's runtime speed. While both \ultraverse and QBS can transpile application code into the SQL equivalent, they have different goals (and thus use different techniques). QBS optimistically searches for any application code snippets that can be converted into more efficient SQL code (e.g., application code's nested for-loop repetitively calling SQL queries can be converted into a single JOIN query). However, QBS is not designed to convert the entire application code. For example, QBS cannot handle the code snippet that has dynamism, such as dynamic data or function pointers. \ultraverse handles them by dynamic symbolic execution and dynamic symbol-spawning. QBS optimistically finds the code portions to optimize, while \ultraverse wholely converts the given code. Thus, QBS and \ultraverse complement each other. For example, once \ultraverse soundly converts application-level transactions into the equivalent SQL procedure, QBS can further optimize it by transforming multiple SQL statements within the SQL procedure into a single JOIN query. ByePy~\cite{byepy} is a Python-to-SQL compiler effectively reducing the communication cost between Python and DBMS. However, unlike \ultraverse, ByePy does not address how to correctly translate an application's dynamic variable type casting, dynamic control flow jumps, and undeterministic return values of APIs. Acorn~\cite{acorn} is designed to improve large-scale dataset analytics by aggressively caching and reusing the intermediate results of query executions. In particular, Acorn's UDF (user-defined function) compiler leverages symbolic execution to convert UDFs into semantically equivalent Spark SQL expressions. However, unlike \ultraverse's SQL transpiler that uses dynamic symbolic execution, Acorn cannot comprehensively capture and incorporate the dynamic behaviors of an application (e.g. dynamic type coercion, dynamic control flow targets, undeterministic APIs) into statically typed SQL code.
\looseness=-1

\section{Conclusion}
\label{sec:conclusion}
\ultraverse efficiently conducts \textit{what-if} analysis for database-using applications, without necessitating manual modifications to legacy application code. Based on abstract syntax trees and dynamic symbolic execution, \ultraverse converts application transaction functions into equivalent SQL \texttt{PROCEDURE}-based application code. Then, \ultraverse incorporates novel DBMS techniques such as column-wise \& row-wise query dependency analysis, and Hash-jumper, significantly enhancing the speed of application-level \textit{what-if} analysis. \looseness=-1




\section*{Acknowledgments}
This work is supported by JSPS Kakenhi JP23K17456, JP23K25157, JP23K28096, and JST CREST JPMJCR22M2. We thank Dr. James Mickens (Harvard Unviersity) for his insights on application-level semantics. We also thank Dr. David Derler (DFINITY) for his comments on the design and analysis of \ultraverse's table hash algorithm. 


\bibliographystyle{abbrv}
\bibliography{110-bibfile} 

\clearpage
\appendix
\normalsize

\onecolumn

\section{The Column-wise \& Row-wise Read/Write Set Policy}
\label{appendix:set-policy}

\begin{table}[h!]
\small
\centering
\noindent
\renewcommand\thetable{A}
\resizebox{\linewidth}{!}{
\begin{tabular}{|l||l|}
\hline
\textbf{{Query Type}} & \textbf{{Column-wise Read \& Write Set Policy}} \\
\hline
\hline
\textbf{\texttt{CREATE} $||$ \texttt{ALTER}} & $R_c=$ \{ \texttt{\_S.`tablename'} (i.e., the DB schema-change-monitoring virtual table's column mapped to this table) \\
\textbf{\texttt{TABLE}}&\textcolor{white}{$R_c=$ \{ } + all \texttt{\_S.`sourceTableName'} of external table(s) referencing with \texttt{FOREIGN KEY}(s) \}\\
& $W_c=$ \{ \texttt{\_S.`tablename'} \}\\
\hline
\textbf{\texttt{DROP} $||$ \texttt{TRUNCATE}} & $R_c=$ \{ \texttt{\_S.`tablename'}  \}\\
\textbf{\texttt{TABLE}}&$W_c=$ \{ \texttt{\_S.`tablename'} \}\\
\hline
\textbf{\texttt{CREATE (OR}} & $R_c=$ \{ \texttt{\_S.`viewname'} + any \texttt{\_S.`sourceViewName'} or \texttt{\_S.`sourceTableName'} externally referencing \}\\
\textbf{\texttt{REPLACE) VIEW}}&$W_c=$ \{ \texttt{\_S.`viewname'} \}\\
\hline
\textbf{\texttt{DROP VIEW}} & $R_c=$ \{ \texttt{\_S.`viewname'} \}\\
& $W_c=$ \{ \texttt{\_S.`viewname'} \}\\
\hline
\textbf{\texttt{CREATE}} $||$ \textbf{\texttt{DROP}} & $R_r=$ \{ \texttt{\_S.`targetname'} (where \texttt{targetname} is the name of the target \texttt{PROCEDURE}/\texttt{FUNCTION}/\texttt{TRIGGER}) \\
\textbf{\texttt{PROCEDURE/TRIGGER}} & $W_r=$ \{ \texttt{\_S.`targetname'} \}\\
\textbf{\texttt{/FUNCTION}} &\\
\hline
\textbf{\texttt{CREATE} \texttt{CONSTRAINT}} & $R_c=$ \{ \texttt{\_S.`tablename'} of its target table + any \texttt{\_S.`sourceTableName'} referencing with \texttt{FOREIGN KEY}(s)\}\\
\textbf{/\texttt{INDEX}}& $W_c=$ \{ \texttt{\_S.`tablename'} of its target table \}\\
\hline
\textbf{\texttt{SELECT}} & $R_c=$ \{ \texttt{\_S.`tablename'} + the columns of the tables (or views) this query's \texttt{SELECT} or \texttt{WHERE} clause accesses\\
&\textcolor{white}{$R_c=$ \{ } + columns of external tables (or views) if this query uses a \texttt{FOREIGN} \texttt{KEY} referencing them\\
&\textcolor{white}{$R_c=$ \{ } + Union of the $R_c$ of this query's inner sub-queries \}  , $W_c=\{\}$ \\
\hline
\textbf{\texttt{INSERT}} & $R_c=$ \{ \texttt{\_S.`tablename'} + the union of the $R_c$ of this query's inner sub-queries + this table's primary key column\\
&\textcolor{white}{$R_c=$ \{ } + the columns of external tables (or views) if this query uses a \texttt{FOREIGN} \texttt{KEY} referencing them \}\\
& $W_c=$ \{ All columns of the target table (or view) this query inserts into \}\\
\hline
\textbf{\texttt{UPDATE}} $||$& $R_c=$ \{ \texttt{\_S.`tablename'} + the union of the $R_c$ of this query's inner sub-queries\\
\textbf{\texttt{DELETE}}&\textcolor{white}{$R_c=$ \} }+ the columns of the target table (or view) this query reads\\
&\textcolor{white}{$R_c=$ \} }+ the columns of external tables (or views) if this query uses a \texttt{FOREIGN} \texttt{KEY} referencing them \\
&\textcolor{white}{$R_c=$ \{ }+ the columns of the tables (or views) read in its \texttt{WHERE} clause \}\\
& $W_c=$ \{ Either the specific updated columns or all deleted columns of the target table (or view)\\
&\textcolor{white}{$R_c=$ \{ } + all external tables' \texttt{FOREIGN KEY} columns that reference this target table's updated/deleted column \}\\
\hline
\textbf{\texttt{BEGIN TRANSACTION} $||$} & $R_c=$ \{ The union of the $R_c$ of all queries within this transaction or procedure \}\\
\textbf{\texttt{CALL PROCEDURE}} &$W_c=$ \{ The union of the $W_c$ of all queries within this transaction or procedure \}\\
\hline
\textbf{\texttt{TRIGGER}-ing queries}&$R_c=$ \{ Add \texttt{\_S.`triggername'} + $R_r$ of its associated \texttt{TRIGGER}-ed query (defined by \texttt{CREATE} \texttt{TRIGGER}) \}\\
&$W_c=$ \{ Add $W_r$ of its associated \texttt{TRIGGER}-ed query \}\\
\hline
\end{tabular}
}
\caption{\ultraverse's policy for generating a column-wise read set ($R_c$) and a write set ($W_c$) for each type of SQL query.}

\label{tab:set-policy}
\end{table}

\newpage

\begin{table*}[h!]
\small
\centering
\noindent
\renewcommand\thetable{B}
\resizebox{\linewidth}{!}{
\begin{tabular}{|l||l|}
\hline
\textbf{{Query Type}} & \textbf{{Row-wise Read \& Write Set Policy }} \\
\hline
\hline
\textbf{\texttt{CREATE} $||$ \texttt{ALTER}} & $R_r=$ \{ $\langle \texttt{\_S.`tablename'} : * \rangle$ \\
\textbf{\texttt{TABLE}}& \textcolor{white}{$R_r=$ \{ } + all $\langle \texttt{\_S.`sourceTableName'} : * \rangle$ of external table(s) referencing with \texttt{FOREIGN} \texttt{KEY}(s) \}
\\
&$W_r=$ \{ $\langle \texttt{\_S.`tablename'} : * \rangle$ \}\\
\hline
\textbf{\texttt{DROP} $||$ \texttt{TRUNCATE}} & $R_r=$ \{ $\langle \texttt{\_S.`tablename'} : * \rangle$ \\
\textbf{\texttt{TABLE}}& $W_r=$ \{ $\langle \texttt{\_S.`tablename'} : * \rangle$ \}\\
\hline
\textbf{\texttt{CREATE (OR}} & $R_r=$ \{ $\langle \texttt{\_S.`viewname'} : * \rangle$ + all $\langle \texttt{\_S.`sourceViewName/sourceTableName'} : * \rangle$ \}\\
\textbf{\texttt{REPLACE) VIEW}}&$W_r=$ \{ $\langle \texttt{\_S.`viewname'} : * \rangle$ \}\\
\hline
\textbf{\texttt{DROP VIEW}} & $R_r=$ $W_r=$ \{ $\langle \texttt{\_S.`viewname'} : * \rangle$ \}\\
\hline
\textbf{\texttt{CREATE}} $||$ \textbf{\texttt{DROP}} & $R_r=$ \{ $\langle \texttt{\_S.`targetname'} : * \rangle$ (where \texttt{targetname} is the name of the target \texttt{PROCEDURE}/\texttt{FUNCTION}/\texttt{TRIGGER}) \\
\textbf{\texttt{PROCEDURE/TRIGGER}} & $W_r=$ \{ $\langle \texttt{\_S.`targetname'} : * \rangle$ \}\\
\textbf{\texttt{/FUNCTION}} &\\
\hline
\textbf{\texttt{CREATE} \texttt{CONSTRAINT}} & $R_r=$ \{ $\langle \texttt{\_S.`tablename'} : * \rangle$ \} \\
/\textbf{\texttt{INDEX}}& $W_r=$ \{ $\langle \texttt{\_S.`tablename'} : * \rangle$ \}\\
\hline
\textbf{\texttt{SELECT}} & $R_r=$ \{ $\langle \texttt{\_S.`tablename'} : * \rangle$ + the RI key(s) of each table that this query reads in its \texttt{WHERE} clause \\
&\textcolor{white}{$R_r=$ \{ } (wildcard if the rows are not specified) \\
&\textcolor{white}{$R_r=$ \{ } + the RI key(s) of those tables referencing with \texttt{FOREIGN} \texttt{KEY} \\
&\textcolor{white}{$R_r=$ \{ } (wildcard if the referenced source column is neither RI column nor alias column) \} \\
&\textcolor{white}{$R_r=$ \{ } + the union of $R_r$ of nested inner queries \} \\ 
&$W_r=\{\}$ \\
\hline
\textbf{\texttt{INSERT}} & $R_r=$ \{ $\langle \texttt{\_S.`tablename'} : * \rangle$  \\
&\textcolor{white}{$R_r=$ \{ } + the union of $R_r$ of nested inner queries\\
&\textcolor{white}{$R_r=$ \{ } + the RI key(s) of those tables referencing with \texttt{FOREIGN} \texttt{KEY} \\
& $W_r=$ \{ The RI key of the inserted row (which should be implicitly identified for each row if \texttt{AUTO\_INCREMENT} is set) \}\\
\hline
\textbf{\texttt{UPDATE}} $||$& $R_r=$ \{ $\langle \texttt{\_S.`tablename'} : * \rangle$ + the union of $R_r$ of nested inner queries\\
\textbf{\texttt{DELETE}}&\textcolor{white}{$R_r=$ \} }+ The RI key(s) of those tables read in its \texttt{WHERE} clause\\
&\textcolor{white}{$R_r=$ \{ } + the RI key(s) of those tables referencing with \texttt{FOREIGN} \texttt{KEY} \\
& $W_r=$ \{ The RI key(s) of the table written according to the constraint specified in its \texttt{WHERE} clause \\
&\textcolor{white}{$R_r=$ \{ } + the RI key(s) of those tables that reference this table by their \texttt{FOREIGN} \texttt{KEY} \\
&\textcolor{white}{$R_r=$ \{ }  (Note: it is those tables that reference this table, not those tables that this table references) \}\\
\hline
\textbf{\texttt{BEGIN TRANSACTION}} & $R_r=$ \{ $\langle \texttt{\_S.`procedurename/functionname'} : * \rangle$ + the union of $R_r$ of all inner queries \}\\
$||$ \textbf{\texttt{CALL PROCEDURE}} &$W_r=$ \{ The union of $W_r$ of all inner queries \}\\
$||$ \textbf{\texttt{FUNCTION()}}&\\
\hline
\textbf{\texttt{TRIGGER}-ing queries}&$R_r=$ \{ Add $\langle \texttt{\_S.`triggername'} : * \rangle$ + $R_r$ of its associated \texttt{TRIGGER}-ed query (defined by \texttt{CREATE} \texttt{TRIGGER}) \}\\
&$W_r=$ \{ Add $W_r$ of its associated \texttt{TRIGGER}-ed query \}\\
\hline
\end{tabular}
}
\caption{\ultraverse's policy for generating a row-wise read set ($R_r$) and a write set ($W_r$) for each type of SQL query.}
\label{tab:row-set-policy}
\end{table*}

\newpage

\section{An example of Row-wise Query Dependency Graph}
\label{appendix:rowwise-dependency-graph}

\begin{figure}[h!]
  \includegraphics[width=\linewidth]{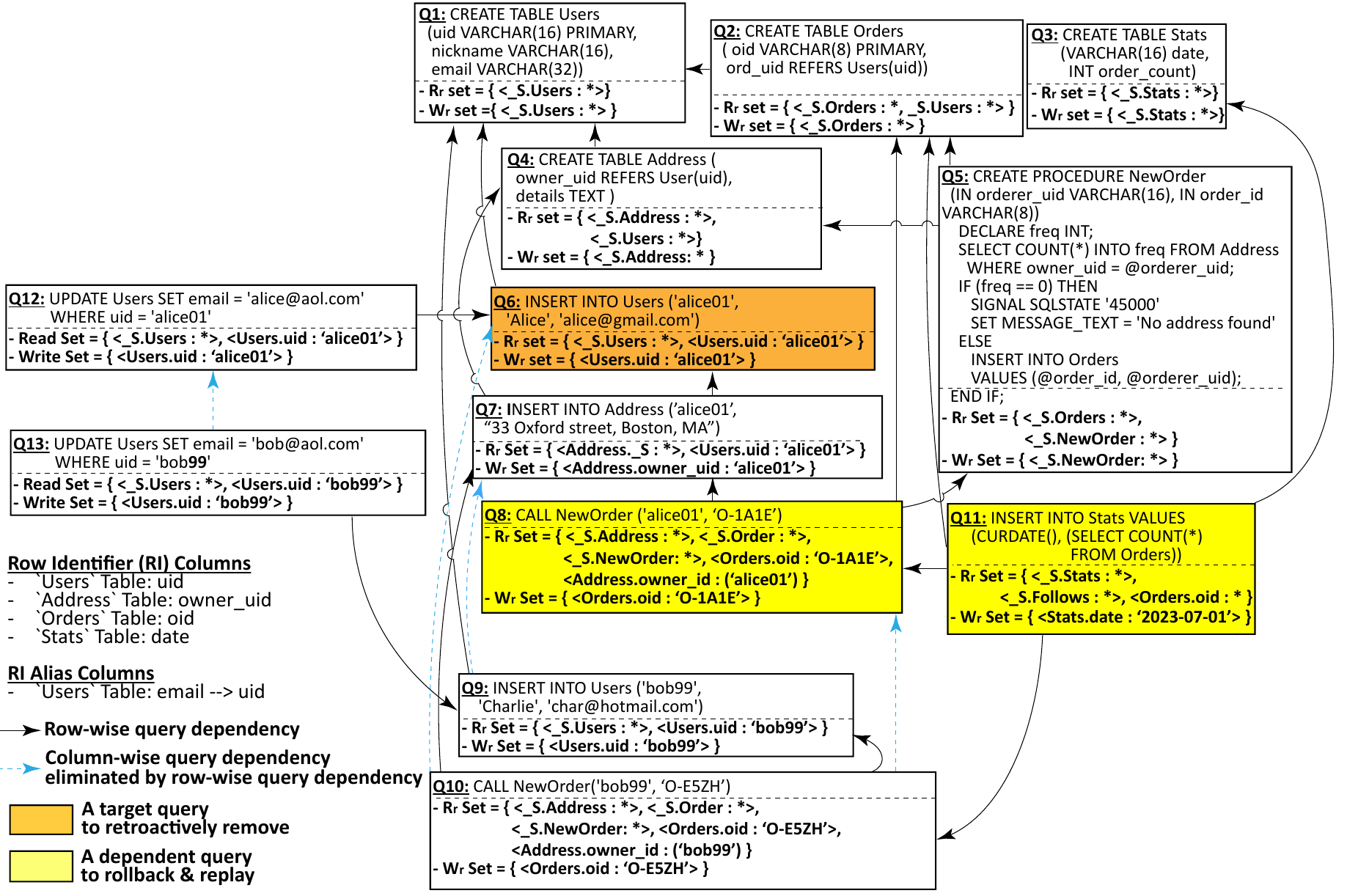}
  \caption{Row-wise query dependency graph for \autoref{fig:dependency-graph-website-columnwise}.}
  \label{fig:dependency-graph-website-rowwise}
\end{figure}

\newpage


\section{Examples of Application-level Dynamism}
\label{appendix:dynamic-application}

\begin{figure*}[h]
\renewcommand\arraystretch{1.2}
   \footnotesize
\centering
\begin{subfigure}[h]{\textwidth}
\begin{footnotesize}
\begin{Verbatim}[frame=single, baselinestretch=0.7, commandchars=\\\{\}]
\textbf{function} dynamic_type \textbf{\textbf{(}}userid, input1, input2, is_string\textbf{\textbf{)}:} 
\{   \textbf{if} \textbf{(}is_string == 1\textbf{)}
       \textbf{SQL_exec(}\textcolor{brown}{`INSERT INTO UserDesc (userid, desc) VALUES (}$\{userid\}\textcolor{brown}{, }$\{input1.concat(input2)\}$\textcolor{brown}{)}`\textbf{)};
    \textbf{else}
       \textbf{SQL_exec(}\textcolor{brown}{`INSERT INTO UserVal (userid, value) VALUES (}$\{userid\}\textcolor{brown}{, }$\{input1 - input2\}$\textcolor{brown}{)}`\textbf{)};
\}
\end{Verbatim}
\end{footnotesize}
\caption{Original application-level transaction.}
\label{fig:dynamic-application-1}
\end{subfigure}
\vspace{1ex}

\begin{subfigure}[h]{\textwidth}
\begin{footnotesize}
\begin{Verbatim}[frame=single, baselinestretch=0.7, commandchars=\\\{\}]
\textbf{function} dynamic_type_augmented \textbf{(}userid, input1, input2, is_string\textbf{)}
\{   Ultraverse_log\textbf{(}`function dynamic_type_augmented (\$\{userid\}, \$\{input1\}, $\{input2\}, $\{is_string\})`\textbf{)};
    \textbf{if} \textbf{(}is_string == 1\textbf{)}
       \textbf{SQL_exec(}\textcolor{brown}{`INSERT INTO UserDesc (userid, desc) VALUES (}$\{userid\}\textcolor{brown}{, }$\{input1.concat(input2)\}$\textcolor{brown}{)}`\textbf{)};
    \textbf{else}
       \textbf{SQL_exec(}\textcolor{brown}{`INSERT INTO UserVal (userid, value) VALUES (}$\{userid\}\textcolor{brown}{, }$\{input1 - input2\}$\textcolor{brown}{)}`\textbf{)};
\}
\end{Verbatim}
\end{footnotesize}
\caption{Augmented application-level transaction.}
\label{fig:dynamic-application-augmented-1}
\end{subfigure}
\vspace{1ex}

\begin{subfigure}[h]{\textwidth}
\begin{footnotesize}
\begin{Verbatim}[frame=single, baselinestretch=0.7, commandchars=\\\{\}]
\textbf{DECLARE PROCEDURE} dynamic_type \textbf{(}\textbf{IN} input1_real \textbf{DOUBLE}, 
  \textbf{IN} input1_str \textbf{VARCHAR}\textbf{(32)}, \textbf{IN} input2_real \textbf{DOUBLE}, 
  \textbf{IN} input2_str \textbf{VARCHAR}\textbf{(32)}, \textbf{IN} is_string_int \textbf{INT}\textbf{)}
\textbf{BEGIN}
   \textbf{IF} \textbf{(}is_string = 1\textbf{)} \textbf{THEN}
      INSERT INTO UserDesc (userid, desc) VALUES (userid, CONCAT(input1_str, input2_str));
   \textbf{ELSE}
      INSERT INTO UserVal (userid, value) VALUES (userid, input1_real - input2_real);
   \textbf{END} \textbf{IF};
\textbf{END}
\end{Verbatim}
\end{footnotesize}
\caption{Transpiled SQL \texttt{PROCEDURE} \textbf{(}with dynamic type resolution\textbf{)}.}
\label{fig:dynamic-application-transpiled-1}
\end{subfigure}

\caption{An example of dynamic type coersion.}
\label{fig:dynamic-code-1}
\end{figure*}

\autoref{fig:dynamic-application-1} shows an example of application-level transaction involving dynamic type coercion, where the input arguments are either a string or numeric type. Their intended type depends on the real-time value of \texttt{is\_string}, and thus it is difficult to determine their type without executing the code. \autoref{fig:dynamic-application-transpiled-1} is a transpiled SQL \texttt{PROCEDURE} that addresses this issue. In \autoref{fig:dynamic-application-transpiled-1}, the arguments cover all possible types of inputs to the transaction (string and double), which are discovered during DSE analysis. This way, the transpiled SQL \texttt{PROCEDURE} projects all possible dynamism of input types in the original application-level transaction into SQL semantics. In addition, notice that the SQL transpiler also generates the augmented application-level transaction (\autoref{fig:dynamic-application-augmented-1}), which is executed instead of \autoref{fig:dynamic-application-1} during the regular operation. The purpose of \autoref{fig:dynamic-application-augmented-1} is to log the type and the input arguments to the application-level transactions (by calling \texttt{Ultraverse\_log}) executed during the regular operation for the purpose of query (or application-level transaction) dependency analysis.

\newpage

\begin{figure*}[t]
\renewcommand\arraystretch{1.2}
   \footnotesize
\centering
\begin{subfigure}[h]{\textwidth}
\begin{footnotesize}
\begin{Verbatim}[frame=single, baselinestretch=0.7, commandchars=\\\{\}]
\textbf{function} dynamic_function_call_augmented \textbf{(}function_name, new_value\textbf{)}
\{   \textbf{return} function_list[function_name]\textbf{(}new_value\textbf{)};
\}
\end{Verbatim}
\end{footnotesize}
\caption{Original application-level transaction.}
\label{fig:dynamic-application-2}
\end{subfigure}
\vspace{1ex}

\begin{subfigure}[h]{\textwidth}
\begin{footnotesize}
\begin{Verbatim}[frame=single, baselinestretch=0.7, commandchars=\\\{\}]
\textbf{function} dynamic_function_call_augmented \textbf{(}function_name, new_value\textbf{)}
\{   Ultraverse_log\textbf{(}`function dynamic_function_call (\$\{function_name\}, $\{new_value\})`\textbf{)};
    \textbf{if} \textbf{(}function_name = \textcolor{brown}{"increment"} || function_name = \textcolor{brown}{"decrement"}\textbf{)}
       \textbf{return} function_list[function_name]\textbf{(}new_value\textbf{)};
    \textbf{else}
       Ultraverse_log\textbf{(}\textcolor{brown}{`[New Exection Path] Found:} 
        \textcolor{brown}{function dynamic_function_call(}\textbf{\$\{}function_name\textbf{\}}\textcolor{brown}{,} \textbf{\$\{}new_value\textbf{\}}\textcolor{brown}{)`}\textbf{)};
\}
\end{Verbatim}
\end{footnotesize}
\caption{Augmented application-level transaction \textbf{(}with delta-updating newly discovered paths\textbf{)}.}
\label{fig:dynamic-application-augmented-2}
\end{subfigure}
\vspace{1ex}

\begin{subfigure}[h]{\textwidth}
\begin{footnotesize}
\begin{Verbatim}[frame=single, baselinestretch=0.7, commandchars=\\\{\}]
\textbf{DECLARE PROCEDURE} dynamic_function_call \textbf{(}\textbf{IN} function_name \textbf{VARCHAR}\textbf{(16)}, \textbf{IN} new_value \textbf{\textbf{IN}T}\textbf{)}
\textbf{BEGIN}
    \textbf{IF} \textbf{(}function_name == \textcolor{brown}{"increment"}\textbf{)} \textbf{THEN}
       \textbf{CALL} increment\textbf{(}new_value, out\textbf{)};
    \textbf{ELIF} \textbf{(}function_name == \textcolor{brown}{"decrement"}\textbf{)} \textbf{THEN}
       \textbf{CALL} decrement\textbf{(}new_value, out\textbf{)};      
    \textbf{END} \textbf{IF};
\textbf{END}
\end{Verbatim}
\end{footnotesize}
\caption{Transpiled SQL \texttt{PROCEDURE}.}
\label{fig:dynamic-application-transpiled-2}
\end{subfigure}

\caption{An example of dynamic control flow targets.}
\label{fig:dynamic-code-2}
\end{figure*}

\autoref{fig:dynamic-application-2} is an example of application-level transaction that involves dynamic control flow targets, where the function name to be called is passed as an input argument. It is difficult to statically determine the target function, because the value of \texttt{function\_name} is resolved dynamically in runtime. \autoref{fig:dynamic-application-transpiled-2} is a transpiled SQL \texttt{PROCEDURE} that addresses this issue. In \autoref{fig:dynamic-application-transpiled-2}, the possible values of the \texttt{function\_name} variable are \texttt{"increment"} and \texttt{"decrement"}, which are discovered during the DSE analysis. Further, \autoref{fig:dynamic-application-augmented-2} has a conditional branch that captures any new function name(s) that is other than the ones (\texttt{"increment"} or \texttt{"decrement"}) discovered during the DSE analysis. If a new value for \texttt{function\_name} is detected during the regular operation, the transaction's input arguments (i.e., \texttt{function\_name} and \texttt{arg}) that lead to triggering this new function gets stashed and used as inputs for delta-DSE analysis, whose newly discovered execution path gets incorporated into the latest transpiled SQL \texttt{PROCEDURE}.

\newpage

\begin{figure*}[t]
\renewcommand\arraystretch{1.2}
   \footnotesize
\centering
\begin{subfigure}[h]{\textwidth}
\begin{footnotesize}
\begin{Verbatim}[frame=single, baselinestretch=0.7, commandchars=\\\{\}]
\textbf{function} external_io \textbf{(}message\textbf{)}
\{   var response = httpRequest.send\textbf{(}message\textbf{)};
    \textbf{if} \textbf{(}response.code == 1\textbf{)}
       \textbf{SQL_exec(}\textcolor{brown}{`INSERT INTO Results (result, log) VALUES ("success", }$\{message\}\textcolor{brown}{)`}\textbf{)};
    \textbf{else}    
       \textbf{SQL_exec(}\textcolor{brown}{`INSERT INTO Results (result, log) VALUES ("fail", }$\{response.error\}\textcolor{brown}{)`}\textbf{)};
\}
\end{Verbatim}
\end{footnotesize}
\caption{Original application-level transaction.}
\label{fig:dynamic-application-3}
\end{subfigure}
\vspace{1ex}

\begin{subfigure}[h]{\textwidth}
\begin{footnotesize}
\begin{Verbatim}[frame=single, baselinestretch=0.7, commandchars=\\\{\}]
\textbf{function} external_io_augmented \textbf{(}message\textbf{)}
\{   Ultraverse_log\textbf{(}\textcolor{brown}{`function external_io(}\textbf{\$\{}message\textbf{\}}\textcolor{brown}{)`}\textbf{)};
    \textbf{var} response = httpRequest.send\textbf{(}message\textbf{)};
    \textbf{if} \textbf{(}response.code == 1\textbf{)}
       \textbf{SQL_exec(}\textcolor{brown}{`INSERT INTO Results (result, log) VALUES ("success", }\textbf{$\{}message\textbf{\}}\textcolor{brown}{)`}\textbf{)};
    \textbf{else}    
       \textbf{SQL_exec(}\textcolor{brown}{`INSERT INTO Results (result, log) VALUES ("fail", }\textbf{$\{}response.error\textbf{\}}\textcolor{brown}{)`}\textbf{)};
\}
\end{Verbatim}
\end{footnotesize}
\caption{Augmented application-level transaction.}
\label{fig:dynamic-application-augmented-3}
\end{subfigure}
\vspace{1ex}

\begin{subfigure}[h]{\textwidth}
\begin{footnotesize}
\begin{Verbatim}[frame=single, baselinestretch=0.7, commandchars=\\\{\}]
\textbf{DECLARE PROCEDURE} external_io \textbf{(}\textbf{IN} message \textbf{VARCHAR}, \textbf{IN} blackbox_symbol_1 \textbf{INT}, 
  \textbf{IN} blackbox_symbol_2 \textbf{VARCHAR}\textbf{)}
\textbf{BEGIN}
   \textbf{IF} \textbf{(}new_symbol_1 == 1\textbf{)} \textbf{THEN}
       INSERT INTO Results (result, log) VALUES (\textcolor{brown}{"success"}, message);
    \textbf{ELSE}    
       INSERT INTO Results (result, log) VALUES (\textcolor{brown}{"fail"}, blackbox_symbol_2);
   \textbf{END} \textbf{IF};
\textbf{END}
\end{Verbatim}
\end{footnotesize}
\caption{Transpiled SQL \texttt{PROCEDURE} \textbf{(}with spawning a new symbol\textbf{)}.} 
\label{fig:dynamic-application-transpiled-3}
\end{subfigure}
\caption{An example of undeterministic blackbox API.}
\label{fig:dynamic-code-3}
\end{figure*}

\autoref{fig:dynamic-application-2} is an example of application-level transaction that involves undeterministic blackbox APIs exchanging messages with an external endpoint. If the external entity is beyond the control of the local \textit{what-if} analyst, it is difficult to statically determine what will be its return value. \autoref{fig:dynamic-application-transpiled-3} is a transpiled SQL \texttt{PROCEDURE} that addresses this issue. In \autoref{fig:dynamic-application-transpiled-3}, the blackbox API's undeterministic return value is treated as a new symbol called \texttt{new\_symbol\_1}. During the \textit{what-if} analysis, the \textit{what-if} analyst can either use the same return value of the blackbox API that was recorded by the augmented application-level transaction (\autoref{fig:dynamic-application-augmented-2}) during the regular operation, or use a different return value to simulate an alternate \textit{what-if} scenario where the remote endpoint responds with a different response code (and a response error).

\newpage

\twocolumn

\section{Benchmarks Analysis by \ultraverse}
\label{appendix:benchmarks}

\subsection{Epinions}
\label{appendix:epinions}

Epinions' dataset and transactions are designed to generate recommendation system networks based on online user reviews. Epinions is comprised of 9 application-level transactions, of which 5 are database-read-only transactions and 4 are database-updating transactions. The transaction dependency graph in \autoref{fig:transaction-dependency-graph-epinions} shows only the later type of transactions.

\begin{figure}[h]
  \centering
  \includegraphics[width=0.5\textwidth]{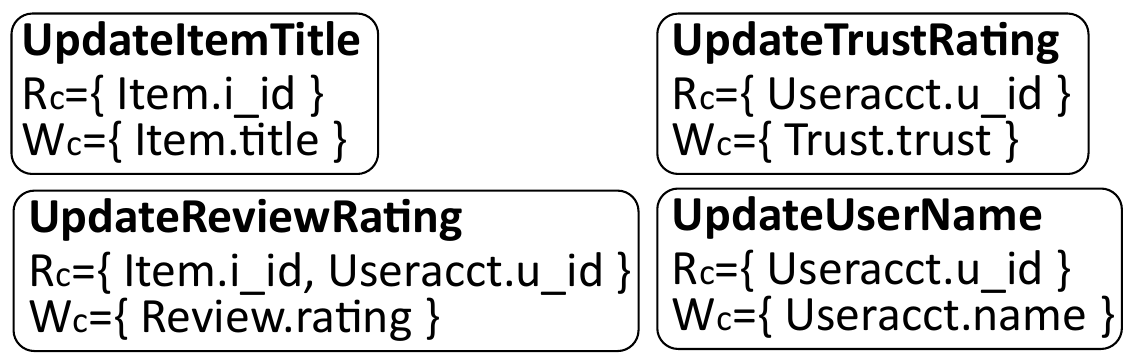}
  \caption{Epinion's column-wise transaction dependency graph (no dependency).}
  \label{fig:transaction-dependency-graph-epinions}
\end{figure}
\noindent The RI (row identification) columns are set as follows:
\begin{itemize}
    \item \texttt{item.i\_id}
    \item \texttt{useracct.u\_id}
    \item \texttt{review.(i\_id, u\_id)}
    \item \texttt{trust.(source\_u\_id, target\_u\_id)}
\end{itemize}

$ $

\noindent No alias columns are set. 

\newpage

\subsection{TATP}
\label{appendix:tatp}

TATP's dataset and transactions are designed for mobile network providers to manage their subscribers. TATP is comprised of 7 application-level transactions, of which 3 are database-read-only transactions and 4 are database-updating transactions. The transaction dependency graph in \autoref{fig:transaction-dependency-graph-tatp} shows only the later type of transactions.

\begin{figure}[h]
  \centering
  \includegraphics[width=1\linewidth]{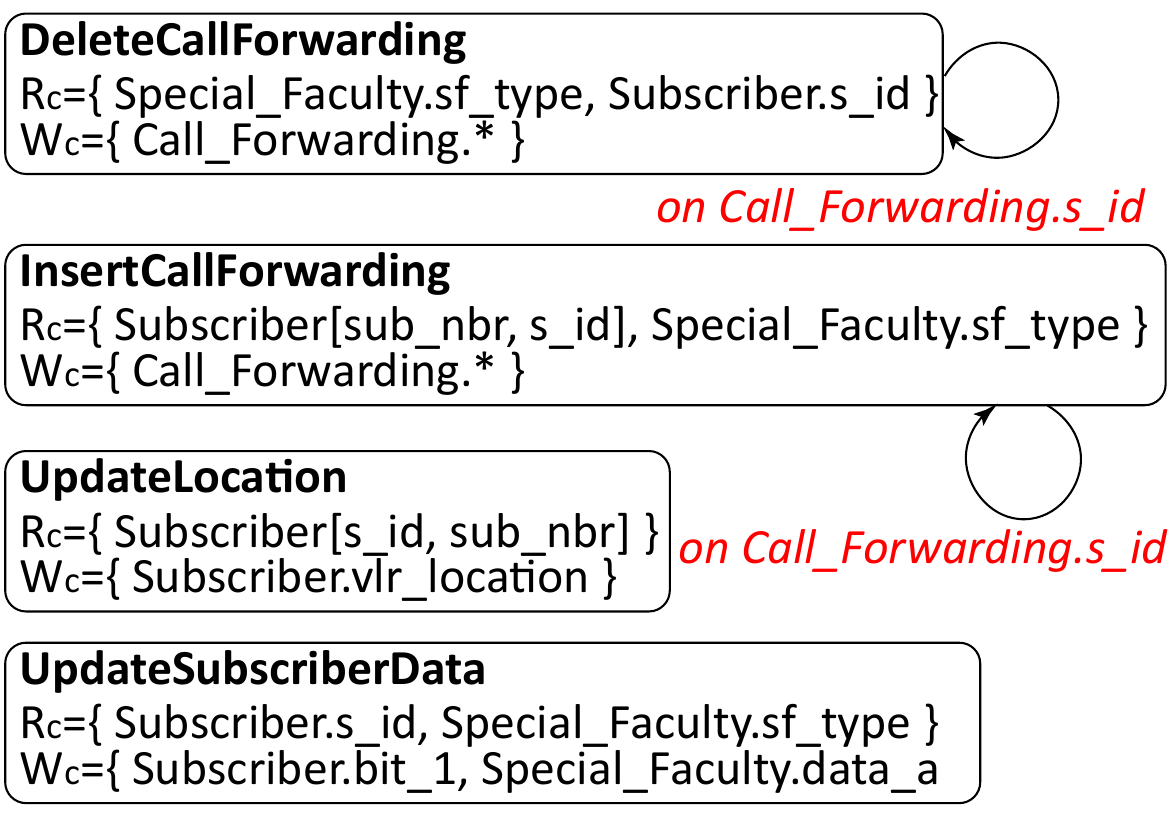}
  \caption{TATP's column-wise1nsaction dependency.}
  \label{fig:transaction-dependency-graph-tatp}
\end{figure}

\noindent The RI columns are set as follows:
\begin{itemize}
    \item \texttt{subscriber.s\_nbr}
    \item \texttt{call\_forwarding.s\_id}
    \item \texttt{special\_facility.s\_id}
\end{itemize}

$ $

\noindent The alias columns are set as follows:
\begin{itemize}
    \item \texttt{subscriber.sub\_nbr} $\rightarrow$ \texttt{subscriber.s\_id}
\end{itemize}

\newpage

\subsection{SEATS}
\label{appendix:seats}

SEAT's dataset and transactions are designed for an online flight ticket reservation system. SEATS is comprised of 7 application-level transactions, of which 3 are database-read-only transactions and 4 are database-updating transactions. The transaction dependency graph in \autoref{fig:transaction-dependency-graph-seats} shows only the later type of transactions.

\begin{figure}[h]
  \centering
  \includegraphics[width=1\linewidth]{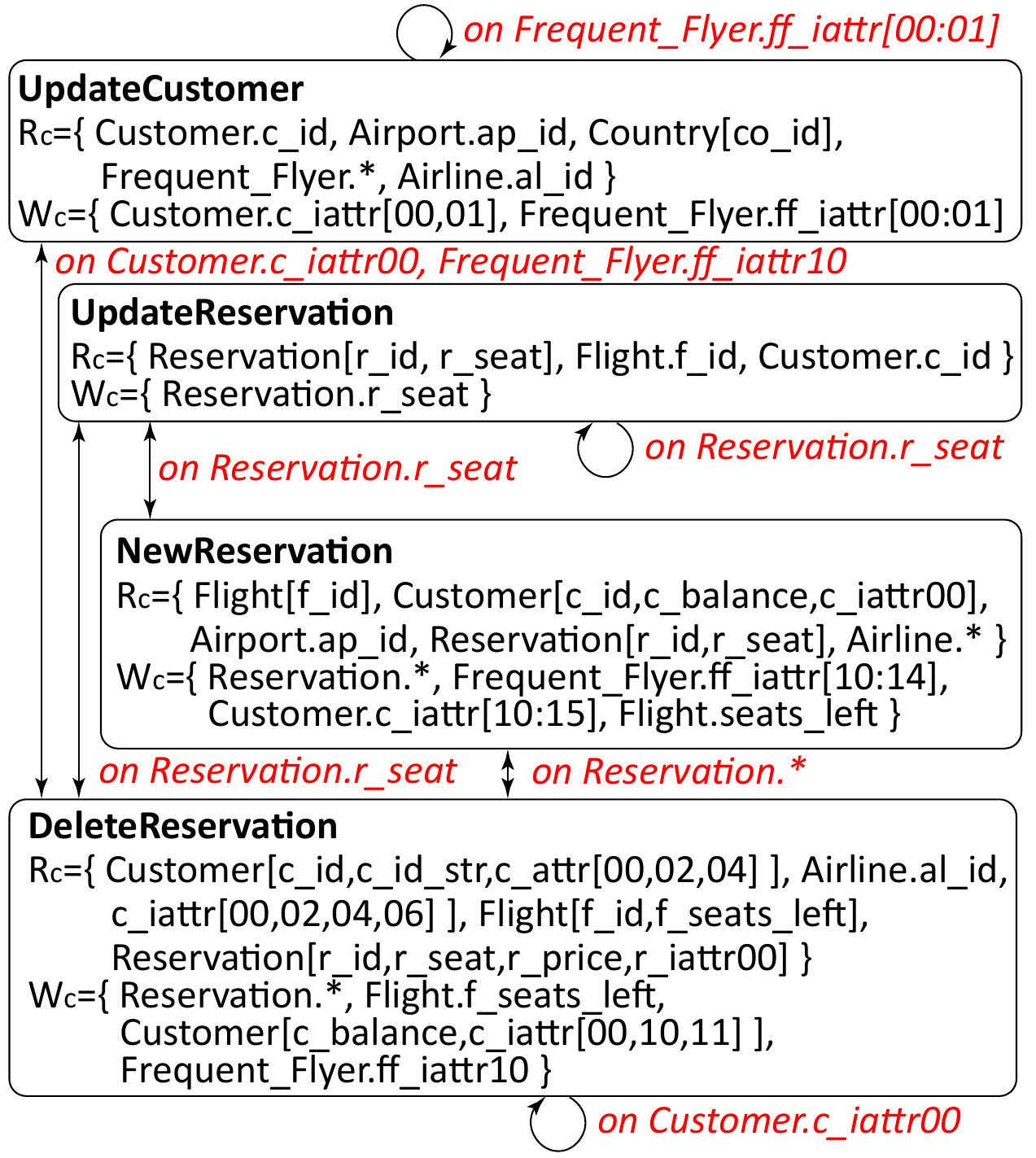}
  \caption{SEAT's column-wise transaction dependency.}
  \label{fig:transaction-dependency-graph-seats}
\end{figure}

\noindent The RI columns are set as follows:
\begin{itemize}
    \item \texttt{customer.C\_ID}
    \item \texttt{flight.F\_ID}
    \item \texttt{frequent\_flyer.FF\_C\_ID}    
    \item \texttt{reservation.(R\_C\_ID, R\_F\_ID)}
    \item \texttt{airport.AP\_ID}
\end{itemize}

$ $

\noindent The alias columns are set as follows:
\begin{itemize}
    \item \texttt{customer.C\_ID\_STR} $\rightarrow$ \texttt{customer.C\_ID}
    \item \texttt{flight.F\_ID} $\rightarrow$ \texttt{flight.F\_AL\_ID}
    \item \texttt{frequent\_flyer.FF\_C\_ID\_STR} $\rightarrow$ \texttt{frequent\_flyer.FF\_C\_ID}
    \item \texttt{airport.AP\_ID} $\rightarrow$ \texttt{airport.AP\_CO\_ID}
\end{itemize}

\newpage

\subsection{TPC-C}
\label{appendix:tpcc}

TPC-C's dataset and transactions are designed to manage product orders shippings for online users in an e-commercial service. TPC-C is comprised of 5 application-level transactions, of which 2 are database-read-only transactions and 3 are database-updating transactions. The transaction dependency graph in \autoref{fig:transaction-dependency-graph-tpcc} shows only the later type of transactions.

\begin{figure}[h]
  \centering
  \includegraphics[width=1\linewidth]{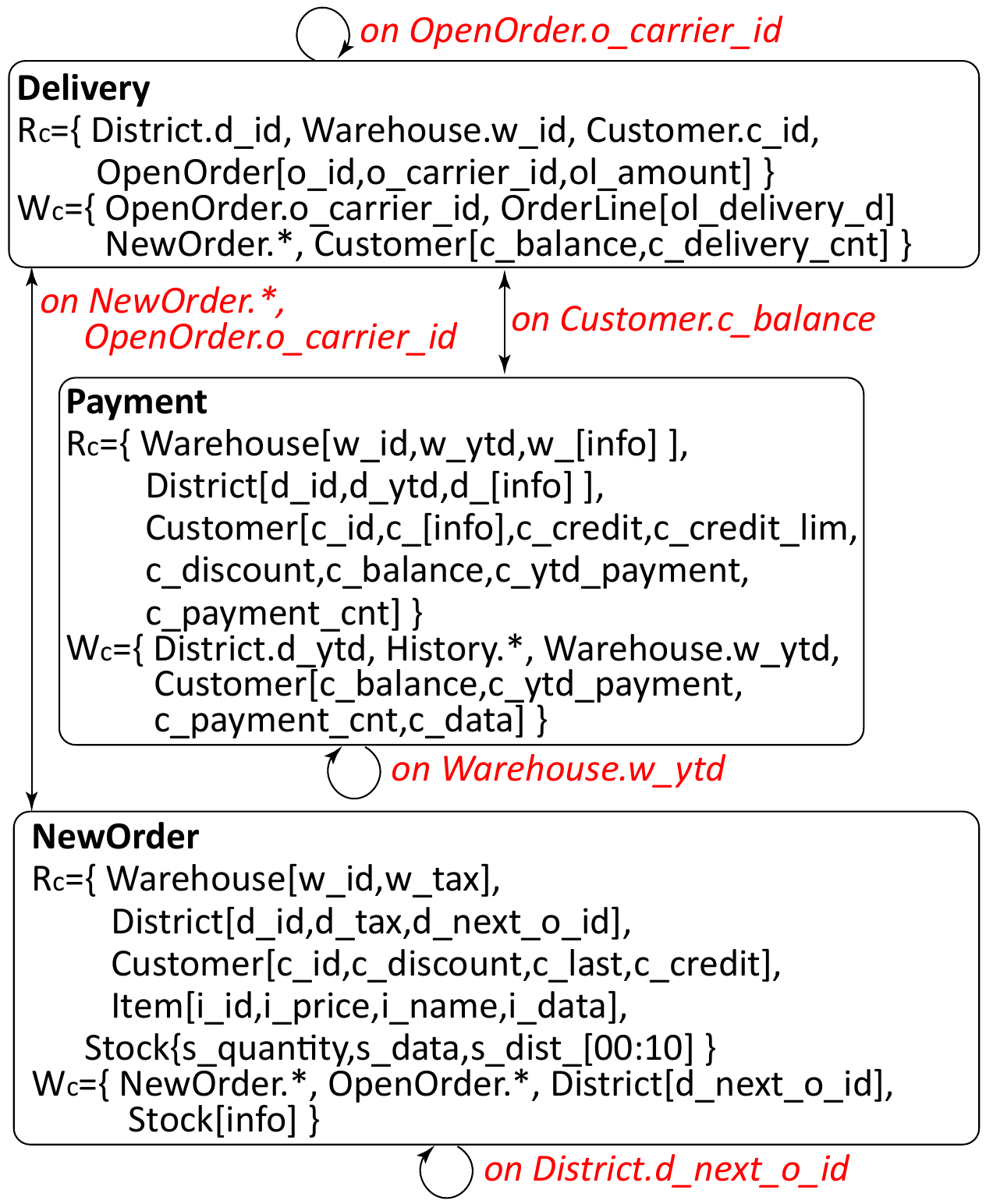}
  \caption{TPC-C's column-wise transaction dependency.}
  \label{fig:transaction-dependency-graph-tpcc}
\end{figure}

\noindent The RI columns are set as follows:
\begin{itemize}
    \item \texttt{warehouse.W\_ID}
    \item \texttt{customer.C\_ID}
    \item \texttt{stocks.S\_W\_ID}
    \item \texttt{order\_line.OL\_W\_ID}
    \item \texttt{district.D\_W\_ID}
    \item \texttt{order.O\_W\_ID}
    \item \texttt{history.H\_C\_W\_ID}
    \item \texttt{item.I\_ID}
\end{itemize}

$ $

\noindent No alias columns are set. 

\newpage

\subsection{AStore}
\label{appendix:astore}

AStore is a ExpressJS-based e-commerce web application where users can purchase products and get the items shipped to their shipping address. AStore is comprised of 61 application-level transactions, of which 41 are database-read-only transactions and 20 transactions are database-updating transactions. The transaction dependency graph in \autoref{fig:transaction-dependency-graph-astore} shows only the later type of transactions.

\begin{figure}[h]
  \centering
  \includegraphics[width=1\linewidth]{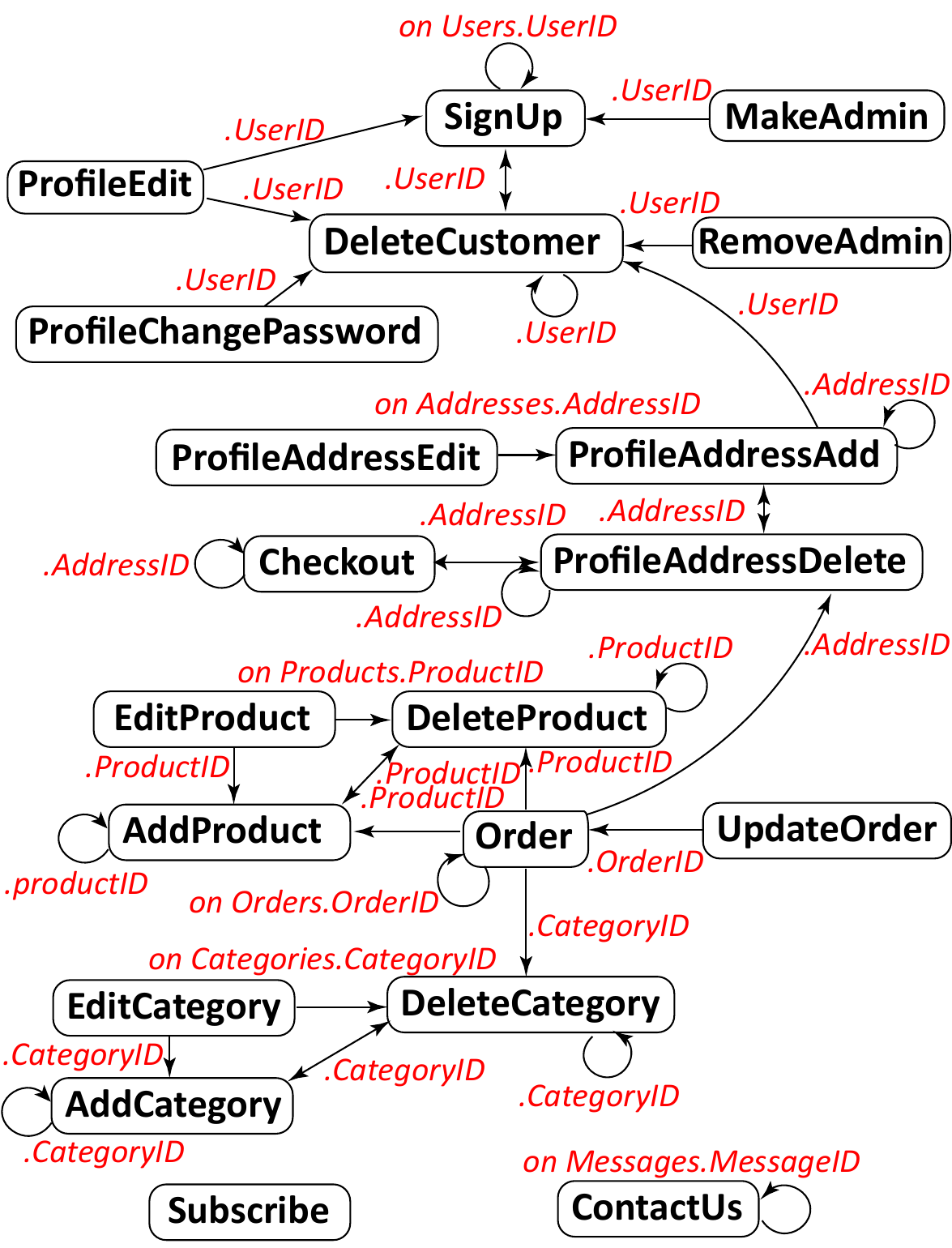}
  \caption{AStore's column-wise transaction dependency.}
  \label{fig:transaction-dependency-graph-astore}
\end{figure}
\noindent The RI columns are set as follows:
\begin{itemize}
    \item \texttt{Users.UserID}
    \item \texttt{Addresses.AddressID}
    \item \texttt{Categories.CategoryID}
    \item \texttt{Products.ProductID}
    \item \texttt{Orders.OrderID}
    \item \texttt{OrderDetails.OrderID}
    \item \texttt{Messages.MessageID}
    \item \texttt{Subscribers.Email}
\end{itemize}

$ $

\noindent No alias columns are set.

\newpage

\newpage

\newpage

\section{Formal Analysis of \ultraverse's Retroactive Operation Techniques}
\label{appendix:analysis}

The formal definition of a retroactive operation is as follows:

\begin{definition}[Retroactive Operation]
Let $\mathbb{D}$ a database and $\mathbb{Q}$ a set of all committed queries $Q_1, Q_2, \cdots Q_{|\mathbb{Q}|}$ where the subscript represents the query's index (i.e., commit order). Let $\mathbb{Q}_{ \langle i, j \rangle}$ be a subset of $\mathbb{Q}$ that contains from i-th to j-th queries in $\mathbb{Q}$, that is $\{Q_i, Q_{i+1}, \cdots, Q_j \}$ (where $i \leq j$). Let $\psi$ be the last query's commit order in $\mathbb{Q}$ (i.e., $|\mathbb{Q}|$). Let $\mathcal{M}: \mathbb{D}, \mathbb{Q} \rightarrow \mathbb{D'}$ be a function that accepts an input database $\mathbb{D}$ and a set of queries $\mathbb{Q}$, executes queries in $\mathbb{Q}$ in ascending order of query indices, and outputs a resulting database $\mathbb{D'}$. Let $\mathcal{M}^{-1}: \mathbb{D}, \mathbb{Q} \rightarrow \mathbb{D'}$ be a function that accepts an input database $\mathbb{D}$ and a set of queries $\mathbb{Q}$, rolls back queries in $\mathbb{Q}$ in descending order of query indices, and outputs a resulting database $\mathbb{D'}$. Given a database $\mathbb{D}$ and a set of committed queries $\mathbb{Q}$, a retroactive operation with a target query $Q'_\tau$ is defined to be a transformation of $\text{ }\mathbb{D}$ to a new state that matches the one generated by the following procedure: 
\begin{enumerate}
\setlength{\itemindent}{1.2em}
    \item Roll back $\mathbb{D}$'s state to commit index $\tau - 1$ by computing $\mathbb{D} := \mathcal{M}^{-1}(\mathbb{D}, \mathbb{Q}_{\langle \tau, \psi \rangle})$.
    \item Depending on the database user's command, do one of the following retroactive operations: 
    \begin{itemize}
        \item In case of retroactively adding $Q'_\tau$, newly execute $Q'_\tau$ by computing $\mathbb{D} := \mathcal{M}(\mathbb{D}, Q'_\tau)$, and then replay all subsequent queries by computing $\mathbb{D} := \mathcal{M}(\mathbb{D}, \mathbb{Q}_{\langle \tau, \psi \rangle})$.
        \item In case of retroactively removing $Q_\tau$, skip replaying $Q_\tau$, and replay all subsequent queries by computing $\mathbb{D} := \mathcal{M}(\mathbb{D}, \mathbb{Q}_{\langle \tau + 1, \psi \rangle})$.
        \item In case of retroactively changing $Q_\tau$ to $Q'_\tau$, newly execute $Q'_\tau$ by computing $\mathbb{D} := \mathcal{M}(\mathbb{D}, Q'_\tau)$, and replay all subsequent queries by computing $\mathbb{D} := \mathcal{M}(\mathbb{D}, \mathbb{Q}_{\langle \tau + 1, \psi \rangle})$.
    \end{itemize}
\end{enumerate}
\end{definition}

The goal of \ultraverse's query analysis is to reduce the number of queries to be rolled back and replayed for a retroactive operation, while preserving its correctness.

\begin{setup}[\ultraverse's Query Analysis] \textcolor{white}{.}

\begin{tabular}{ll}
\textbf{Input}&: $\mathbb{D, Q}, \langle Q'_\tau, \mathit{add | remove | change} \rangle$.\\
\textbf{Output}&: A subset of \textcolor{white}{.}$ \mathbb{Q} $ to be rolled back and replayed.
\end{tabular}
\label{setup:setup}
\end{setup}

Setup~\ref{setup:setup} describes the input and output of \ultraverse's query analysis. The input is $\mathbb{D}$ (a database), $\mathbb{Q}$ (a set of all committed queries), $Q'_\tau$ (a retroactive target query to be added or changed to), and the type of retroactive operation on $Q'_\tau$ (i.e., add, remove, or change it). Note that in case of retroactive removal of the query at the commit index $\tau$, the retroactive target query $Q'_\tau$ in the \textbf{\textit{Input}} is $Q_\tau$. The output is a subset of $\mathbb{Q}$. Rolling back and replaying the output queries result in a correct retroactive operation.

\ultraverse's query analysis is comprised of two components: column-wise query dependency analysis and row-wise query dependency analysis. We will first describe column-wise query dependency analysis and then extend to row-wise query dependency analysis. To show the correctness of performing retroactive operations using query analysis, we first assume that the retroactive operation is either adding or removing a query, and address the case of retroactively changing a query at the end of this section. 

\subsection{Column-wise Query Dependency Analysis}
\label{appendix:columnwise-dependency-analysis}

\begin{terminology}[Query Dependency Analysis] \textcolor{white}{.}

\begin{tabular}{@{}l@{}l}
$\mathbb{D}$&: A given database\\
$\mathbb{Q}$&: A set of all committed queries in $\mathbb{D}$\\
$\mathbb{Q}_{ \langle i, j \rangle}$&: A subset of $\mathbb{Q}$ from the i-th to j-th queries\\
$\bm{Q_n}$&: a query with index $n$ in $\mathbb{Q}$\\
$\bm{\tau}$&: a retroactive target query's index in $\mathbb{Q}$\\
$\bm{Q'_\tau}$&: The retroactive target query to add or change\\
$\bm{Q_n}\triangleright\bm{T}$&: $\bm{Q_n}$ is a \texttt{"CREATE/DROP TRIGGER"} query\\
$\bm{Q_n}\triangleright\bm{P}$&: $\bm{Q_n}$ is a \texttt{"CREATE/DROP PROCEDURE"} query\\
$\bm{R_c(Q_n)}$&: $Q_n$'s column-wise read set\\
$\bm{W_c(Q_n)}$&: $Q_n$'s column-wise write set\\
$\bm{c}$&: a table's column\\
$\bm{Q_n} \rightarrow \bm{Q_m}$&: $Q_n$ depends on $\bm{Q_m}$ \\
$\bm{Q_m} \curvearrowright \bm{Q_n}$&: $Q_m$ is an influencer of $Q_n$\\
\end{tabular}
\label{not:notation1}
\end{terminology}

\begin{definition}[Read/Write Set]
A query $Q_i$'s read set is the set of column(s) that $Q_i$ operates on with read access. $Q_i$'s write set is the set of column(s) that $Q_i$ operates on with write access. 
\end{definition}

For each type of SQL statements, its read \& write sets are determined according to the policies described in \autoref{tab:set-policy}.

Loosely speaking, given $\mathbb{D}$ and $\mathbb{Q}$, we define that $Q_i$ depends on $Q_j$ if some retroactive operation on $Q_i$ \textit{could} change the result of $Q_j$ (i.e., $Q_j$'s return value or the state of the resulting table that $Q_j$ writes to). In this section, when we say query dependency, it always implies column-wise query dependency (discussed in \autoref{subsec:columnwise-dependency-analysis}). We present the formal definition of query dependency in Definition~\ref{def-dependency}.

\begin{definition}[Query Dependency]
Given a database $\mathbb{D}$ and a set of all committed queries $\mathbb{Q}$, one query depends on another if they satisfy Proposition~\ref{prop-dependency-rule1} or \ref{prop-dependency-rule2}.
\label{def-dependency}
\end{definition}

\begin{proposition}
$\exists c ((c \in W_c(Q_m)) \wedge (c \in R_c(Q_n))) \wedge (m < n) \Longrightarrow Q_n \rightarrow Q_m
$
\label{prop-dependency-rule1}
\end{proposition}

Proposition~\ref{prop-dependency-rule1} states that if $Q_n$ reads the table/view's column after $Q_m$ writes to it, then $Q_n$ depends on $Q_m$. Proposition~\ref{prop-dependency-rule1} captures the cases where one query reads the same column that was retroactively modified by some prior query, which can potentially change the result of the latter query due to the changed state of the common column they access. column's state that the later query accesses. 

\begin{proposition}
$(Q_n \rightarrow Q_m) \wedge (Q_m \rightarrow Q_l) \Longrightarrow Q_n \rightarrow Q_l$
\label{prop-dependency-rule2}
\end{proposition}

Proposition~\ref{prop-dependency-rule2} states that if $Q_n$ depends on $Q_m$ and $Q_m$ depends $O_l$, then $Q_n$ also depends on $Q_l$ (transitivity). Proposition~\ref{prop-dependency-rule2} captures the cases where two queries, $Q_n$ and $Q_l$, do not operate on the same column, but there exists some intermediate query $Q_m$ which operates on some same column as each of $Q_n$ and $Q_l$. In such cases, $Q_m$ acts as a data flow bridge between $Q_l$'s column and $Q_n$'s column, and therefore, a retroactive operation on $Q_l$ could change the column's state that $Q_n$ accesses. Therefore, we regards that $Q_n$ depends on $Q_l$ transitively.

\begin{definition}[$\mathbb{I}$]
$\mathbb{I}$ is an intermediate set of all queries that are selected to be rolled back and replayed for a retroactive target query $Q'_\tau$.
\label{def-intermediate-set}
\end{definition}

We define the $\mathbb{I}$ set for three purposes. First, we add the queries dependent on the retroactive target query $Q'_\tau$ to $\mathbb{I}$, as candidate queries to be rolled back and replayed. Second, we further add more queries that need to be rolled back and replayed in order to replay \textit{consulted table(s)} (discussed in \autoref{subsec:rollback-replay}). Third, we remove those queries that do not row-wies depend on the retroactive target query $Q'_\tau$ (discussed in \autoref{subsec:rowwise-analysis}).  

\begin{proposition}
$(Q_i \rightarrow Q'_\tau) \wedge (W_c(Q_i) \neq \emptyset) \Longrightarrow Q_i \in \mathbb{I}$.
\label{prop-i-set}
\end{proposition}

Proposition~\ref{prop-i-set} states if $Q_i$ depends on the retroactive target query $Q'_\tau$ and $Q_i$'s write set is not empty, then $Q_i$ is added to $\mathbb{I}$ (we do not rollback and replay if $Q_i$'s write set is empty, because a read-only query does not change the database's state). Proposition~\ref{prop-i-set} presents our first purpose of using $\mathbb{I}$.

To find the queries to be rolled back and replayed in order to replay \textit{consulted table(s)}, we introduce a new term, \emph{read-then-writer}.

\begin{definition}[Read-then-Writer] 
$\exists c, f ((c \in R_c(Q_i)) \wedge (f = \argmin_{j} ( (j > i) \wedge ( c \in W_c(Q_j) ) ) ) )
\Longrightarrow Q_f	\curvearrowright_c Q_i$
\label{def-read-then-writer}
\end{definition}

Definition~\ref{def-read-then-writer} states that if $Q_j$ is the first query to write to column $c$ after $Q_j$ reads it, then $Q_f$ is defined to be the read-then-writer of $Q_i$ on column $c$.

\begin{proposition}
$\exists i,f ((Q_i \in \mathbb{I}) \wedge (Q_f \curvearrowright_c Q_i)) \Longrightarrow Q_f \in \mathbb{I}$
\label{prop-read-then-writer1}
\end{proposition}

\begin{proposition}
$\exists i,j,c ( (i < j) \wedge (Q_i \in \mathbb{I}) \wedge (c \in R_c(Q_i)) \wedge (c \in W_c(Q_j)) \Longrightarrow Q_j \in \mathbb{I})$
\label{prop-read-then-writer2}
\end{proposition}

Proposition~\ref{prop-read-then-writer1} states that if $Q_i$ is a read-then-writer of some replay query in $\mathbb{I}$, then $Q_i$ is added to $\mathbb{I}$. 

Proposition~\ref{prop-read-then-writer2} states that if $Q_j$ writes to column $c$ at any time after some replay query $Q_i$ reads column $c$, then $Q_j$ is also added to the replay set $\mathbb{I}$. 

Proposition~\ref{prop-read-then-writer1} and \ref{prop-read-then-writer2} ensure to replay all \textit{consulted tables}.

Once Proposition~\ref{prop-dependency-rule1}, \ref{prop-dependency-rule2}, \ref{prop-i-set}, \ref{prop-read-then-writer1}, and \ref{prop-read-then-writer2} complete repetition until there are no more new queries to be added to $\mathbb{I}$, column-wise query dependency analysis is complete.

\begin{theorem}
For a retroactive operation for adding or removing the target query $Q'_\tau$ based on the column-wise query dependency analysis, it is sufficient to do the following: \textit{(i)} rollback the queries that belong to $\mathbb{I}$; \textit{(ii)} either execute $Q'_\tau$ (in case of retroactively adding $Q'_\tau$) or roll back $Q'_\tau$ (in case of retroactively removing $Q'_\tau$); \textit{(iii)} replay all queries in $\mathbb{I}$.
\label{theorem-column-dependency}
\end{theorem}

\begin{proof}
 Let the database state after the retroactive operation of adding the target query $Q'_\tau$ be $\mathbb{D'} = \mathcal{M}(\mathcal{M}(\mathcal{M}^{-1}({\mathbb{D}, \mathbb{Q_{\langle \tau, \psi \rangle}}}), Q'_\tau), \mathbb{Q_{\langle \tau, \psi \rangle}})$. Let $b_i$ be the $i$-th oldest query index that satisfies the following: $(b_i > \tau) \wedge (Q_{b_i} \not\in \mathbb{I})$. For example, $Q_{b_1}$ is the oldest query in $\mathbb{Q_{\langle \tau, \psi \rangle}}$ that does not belong to $\mathbb{I}$, and $Q_{b_2}$ is the second-oldest query in $\mathbb{Q_{\langle \tau, \psi \rangle}}$ that does not belong to $\mathbb{I}$. Note that every query that does not belong to $\mathbb{I}$ also does not depend on any query in $\mathbb{I}$ (otherwise, it should have been put into $\mathbb{I}$ by Proposition~\ref{prop-dependency-rule1}, \ref{prop-dependency-rule2}, \ref{prop-i-set}, \ref{prop-read-then-writer1}, and \ref{prop-read-then-writer2}). We prove Theorem~\ref{theorem-column-dependency} by finite induction.

$ $

\textbf{Case $\bm{Q_{b_1}}$:} 
Let $\mathbb{D}_{b_1} = \mathcal{M}(\mathcal{M}(\mathcal{M}^{-1}(\mathbb{D}, \mathbb{Q_{\langle \tau, \psi \rangle}} - \{ Q_{b_1} \}), Q'_\tau), \mathbb{Q_{\langle \tau, \psi \rangle}} - \{ Q_{b_1} \})$, which is equivalent to rolling back all queries in $\mathbb{Q_{\langle \tau, \psi \rangle}}$ except for $Q_{b_1}$, executing $Q'_\tau$, and replaying all queries in $\mathbb{Q_{\langle \tau, \psi \rangle}}$ except for $Q_{b_1}$. 

\begin{lemma}

$\mathbb{D'} = \mathbb{D}_{b_1}$, which is equivalent to:

$\mathcal{M}(\mathcal{M}(\mathcal{M}^{-1}({\mathbb{D}, \mathbb{Q_{\langle \tau, \psi \rangle}}}), Q'_\tau), \mathbb{Q_{\langle \tau, \psi \rangle}}) $

$ = \mathcal{M}(\mathcal{M}(\mathcal{M}^{-1}(\mathbb{D}, \mathbb{Q_{\langle \tau, \psi \rangle}} - \{ Q_{b_1} \}), Q'_\tau), \mathbb{Q_{\langle \tau, \psi \rangle}} - \{ Q_{b_1} \})$.
\label{lemma-1}
\end{lemma}

The definition of $Q_{b_1}$ implies that in $\mathbb{Q}_{\langle \tau, \psi \rangle}$, any query committed before $Q_{b_1}$ belongs to $\mathbb{I}$. 
Proposition~\ref{prop-dependency-rule2} and  \ref{prop-i-set} guarantee that $Q_{b_1}$ does not depend on any query in $\mathbb{I}$ (because otherwise, $Q_{b_1}$ would have been put into $\mathbb{I}$). This means that the results of $Q_{b_1}$ (i.e., the resulting state of its write set columns) will not be affected by the retroactive operation. Therefore, $Q_{b_1}$ needs not be rolled back \& replayed while generating $\mathbb{D}'$. Thus, Lemma~\ref{lemma-1} is true. 

$ $

\textbf{Case $\bm{Q_{b_2}}$:} 
Let $\mathbb{Q}_{b_1}$ be the set of rolled back \& replayed queries for generating $\mathbb{D}_{b_1}$. Let $\mathbb{D}_{b_2} = \mathcal{M}(\mathcal{M}(\mathcal{M}^{-1}(\mathbb{D}, \mathbb{Q_{\langle \tau, \psi \rangle}} - \{ Q_{b_1}, Q_{b_2} \}), Q'_\tau),\\ \mathbb{Q_{\langle \tau, \psi \rangle}} - \{ Q_{b_1}, Q_{b_2} \})$, which is equivalent to rolling back all queries in $\mathbb{Q_{\langle \tau, \psi \rangle}}$ except for $\{ Q_{b_1}, Q_{b_2} \}$, executing $Q'_\tau$, and replaying all queries in $\mathbb{Q_{\langle \tau, \psi \rangle}}$ except for $\{ Q_{b_1}, Q_{b_2} \}$. 

\begin{lemma}
$\mathbb{D}_{b_1} = \mathbb{D}_{b_2}$, which is equivalent to:

$ \mathcal{M}(\mathcal{M}(\mathcal{M}^{-1}(\mathbb{D}, \mathbb{Q_{\langle \tau, \psi \rangle}} - \{ Q_{b_1} \}), Q'_\tau), \mathbb{Q_{\langle \tau, \psi \rangle}} - \{ Q_{b_1} \})$.

$ = \mathcal{M}(\mathcal{M}(\mathcal{M}^{-1}(\mathbb{D}, \mathbb{Q_{\langle \tau, \psi \rangle}} - \{ Q_{b_1}, Q_{b_2}  \}), Q'_\tau), \mathbb{Q_{\langle \tau, \psi \rangle}} - \{ Q_{b_1}, Q_{b_2} \})$.
\label{lemma-2}
\end{lemma}

The definition of $Q_{b_2}$ implies that in $\mathbb{Q}_{\langle \tau, \psi \rangle}$, any query committed before $Q_{b_1}$ belongs to $\mathbb{I} \cup \{ Q_{b_1} \}$. But in case of $\mathbb{D}_{b_1}$, $\mathbb{Q}_{b_1}$ does not contain $Q_{b_1}$, and thus in $\mathbb{Q}_{b_1}$, any query committed before $Q_{b_1}$ belongs to $\mathbb{I}$. Then, based on the same reasoning used for proving Lemma~\ref{lemma-1}, the results of $Q_{b_2}$ are not affected by the retroactive operation and its results are to be accessed only by those queries committed after $Q_{b_2}$. Thus, $Q_{b_2}$ needs not be rolled back \& replayed while generating $\mathbb{D}_{b_1}$. Thus, Lemma~\ref{lemma-2} is true.

\begin{flushleft}
\textcolor{white}{....}\textbf{Case $\bm{Q_{b_{\psi - |\mathbb{I}|}}}$:} Let $\mathbb{Q}_{b_{\psi - |\mathbb{I}| - 1}}$ be the set of rolled back \& replayed queries for generating $\mathbb{D}_{b_{\psi - |\mathbb{I}| - 1}}$. Let $\mathbb{D}_{b_{\psi - |\mathbb{I}|}} = \mathcal{M}(\mathcal{M}(\mathcal{M}^{-1}(\mathbb{D}, \mathbb{Q_{\langle \tau, \psi \rangle}} - \{ Q_{b_1}, Q_{b_2}, \cdots Q_{b_{\psi - |\mathbb{I}|}} \}), Q'_\tau), \mathbb{Q_{\langle \tau, \psi \rangle}} - \{ Q_{b_1}, Q_{b_2}, \cdots Q_{b_{\psi - |\mathbb{I}|}} \})$, which is equivalent to rolling back all queries in $\mathbb{Q_{\langle \tau, \psi \rangle}}$ except for \\ $\{ Q_{b_1}, Q_{b_2}, \cdots Q_{b_{\psi - |\mathbb{I}|}} \}$, executing $Q'_\tau$, and replaying all queries in $\mathbb{Q_{\langle \tau, \psi \rangle}}$ except for $\{ Q_{b_1}, Q_{b_2}, \cdots Q_{b_{\psi - |\mathbb{I}|}} \}$. 
\end{flushleft}

\begin{lemma}
$\mathbb{D}_{b_{\psi - |\mathbb{I}| - 1}} = \mathbb{D}_{b_{\psi - |\mathbb{I}|}}$, which is equivalent to:

\begin{flushleft}
\noindent$\mathcal{M}(\mathcal{M}(\mathcal{M}^{-1}(\mathbb{D},\mathbb{Q_{\langle \tau, \psi \rangle}} - \{ Q_{b_1}, Q_{b_2}, \cdots Q_{b_{\psi - |\mathbb{I}| - 1}} \}), Q'_\tau),$
\textcolor{white}{...}$\mathbb{Q_{\langle \tau, \psi \rangle}} - \{ Q_{b_1}, Q_{b_2}, \cdots Q_{b_{\psi - |\mathbb{I}| - 1}} \})$
$ = \mathcal{M}(\mathcal{M}(\mathcal{M}^{-1}(\mathbb{D}, \mathbb{Q_{\langle \tau, \psi \rangle}} - \{ Q_{b_1}, Q_{b_2}, \cdots Q_{b_{\psi - |\mathbb{I}|}} \}), Q'_\tau),$
\textcolor{white}{...}$\mathbb{Q_{\langle \tau, \psi \rangle}} - \{ Q_{b_1}, Q_{b_2}, \cdots Q_{b_{\psi - |\mathbb{I}|}} \})$
\end{flushleft}
\label{lemma-3}
\end{lemma}

The definition of $Q_{b_{\psi - |\mathbb{I}|}}$ implies that in $\mathbb{Q}_{\langle \tau, \psi \rangle}$, any query committed before $Q_{b_{\psi - |\mathbb{I}|}}$ belongs to $\mathbb{I} \cup \{ Q_{b_1}, Q_{b_2}, \cdots Q_{b_{\psi - |\mathbb{I}|}} \}$. But in case of $\mathbb{D}_{b_{\psi - |\mathbb{I}| - 1}}$, $\mathbb{Q}_{b_{\psi - |\mathbb{I}| - 1}}$ does not contain $\{ Q_{b_1}, Q_{b_2}, \cdots Q_{b_{\psi - |\mathbb{I}|}} \}$, and thus in $\mathbb{Q}_{b_{{\psi - |\mathbb{I}| - 1}}}$, any query committed before $Q_{b_{{\psi - |\mathbb{I}|}}}$ belongs to $\mathbb{I}$. Then, based on the same reasoning used for proving Lemma~\ref{lemma-1}, the results of $Q_{b_{{\psi - |\mathbb{I}|}}}$ are the same as before the retroactive operation and its results are to be used only by those queries committed after $Q_{b_{{\psi - |\mathbb{I}|}}}$. Thus, $Q_{b_{{\psi - |\mathbb{I}|}}}$ needs not be rolled back \& replayed while generating $\mathbb{D}_{b_{\psi - |\mathbb{I}| - 1}}$. Thus, Lemma~\ref{lemma-3} is true.

$ $

According to Lemma~\ref{lemma-1}, \ref{lemma-2} and \ref{lemma-3}, 

$\mathbb{D'} = \mathbb{D}_{b_1} = \mathbb{D}_{b_2} = \cdots = \mathbb{D}_{b_{\psi - |\mathbb{I}|}}$

$ = \mathcal{M}(\mathcal{M}(\mathcal{M}^{-1}(\mathbb{D}, \mathbb{I}), Q'_\tau), \mathbb{I})$

$ $

Therefore, during the retroactive operation of adding the target query $Q'_\tau$, all queries that do not belong to $\mathbb{I}$ need not be rolled back \& replayed, and the resulting database's state is still consistent. 

If the retroactive operation is removing $Q'_\tau$, then the resulting database's state is 

$\mathbb{D'} = \mathcal{M}(\mathcal{M}^{-1}({\mathbb{D}, \mathbb{Q}_{\langle \tau, \psi \rangle}}), \mathbb{Q}_{\langle \tau + 1, \psi \rangle})$. 

Then, a similar induction proof used for Lemma~\ref{lemma-1}, \ref{lemma-2}, \ref{lemma-3} can be applied to derive the following:

$\mathcal{M}(\mathcal{M}^{-1}({\mathbb{D}, \mathbb{Q}_{\langle \tau, \psi \rangle}}), \mathbb{Q}_{\langle \tau + 1, \psi \rangle})$ 

$ = \mathcal{M}(\mathcal{M}^{-1}(\mathbb{D}, \mathbb{Q}_{\langle \tau, \psi \rangle} - \{ Q_{b_1} \}), \mathbb{Q}_{\langle \tau + 1, \psi \rangle} - \{ Q_{b_1} \})$

$ = \mathcal{M}(\mathcal{M}^{-1}(\mathbb{D}, \mathbb{Q}_{\langle \tau, \psi \rangle} - \{ Q_{b_1}, Q_{b_2}  \}), \mathbb{Q}_{\langle \tau + 1, \psi \rangle} - \{ Q_{b_1}, Q_{b_2} \})$

$\cdots$

$ = \mathcal{M}(\mathcal{M}^{-1}(\mathbb{D}, \mathbb{Q}_{\langle \tau, \psi \rangle} - \{ Q_{b_1}, Q_{b_2}, \cdots Q_{b_{\psi - |\mathbb{I}|}} \}), \mathbb{Q}_{\langle \tau + 1, \psi \rangle} $

$- \{ Q_{b_1}, Q_{b_2}, \cdots Q_{b_{\psi - |\mathbb{I}|}} \})$

$ = \mathcal{M}(\mathcal{M}^{-1}(\mathcal{M}^{-1}(\mathbb{D}, \mathbb{I}), Q'_\tau), \mathbb{I})$

\end{proof}

\subsection{Row-wise Query Dependency Analysis}
\label{appendix:rowwise-dependency-analysis} 

Next, we describe row-wise query dependency analysis to further reduce the number of queries in $\mathbb{I}$. First, we present additional notations as described in Terminology~\ref{not:notation2}.

\begin{terminology}[Row-wise Query Dependency Analysis] \textcolor{white}{.}

\begin{tabular}{ll}
$\bm{R_r(Q_n)}$ | $\bm{W_r(Q_n)}$&: $Q_n$'s row-wise read/write sets\\
\end{tabular}
\label{not:notation2}
\end{terminology}

\begin{definition}[Row Identifier (RI) Column]
A row identification (RI) column is the column of a table in a given database $\mathbb{D}$, whose value represent the table row(s) a query accesses. 
\label{def-ri-column}
\end{definition}

\begin{definition}[Row Identifier (RI) Key]
$\langle c: v \rangle$ is a row identifier (RI) key, where $c$ indicates the RI column and $v$ is the value, set, or range of the data record(s) of the RI column. 
\label{def-ri-key}
\end{definition}

\begin{definition}[Row-wise Read/Write Set]
$\bm{R_r(Q_n)}$ and $\bm{W_r(Q_n)}$ are row-wise read/write sets, containing zero or more RI keys as elements. 
\label{def-ri-set}
\end{definition}

\begin{proposition}
$\exists c : ((\langle c: v \rangle \in W_r(Q_m)) \wedge (\langle c: v \rangle \in R_r(Q_n))) \wedge (m < n) \Longrightarrow Q_n \rightarrow Q_m$
\label{prop-row-dependency-rule1}
\end{proposition}

Proposition~\ref{prop-row-dependency-rule1} states that if $Q_n$ reads or writes the table/view's row after $Q_m$ writes to it, then $Q_n$ depends on $Q_m$. Proposition~\ref{prop-row-dependency-rule1} captures the case where one query reads the same row that was retroactively modified by some prior query, which can potentially change the result of the latter query due to the changed state of the common row they access. 

Proposition~\ref{prop-dependency-rule2} from the column-wise dependency analysis applies the same to the row-wise dependency analysis. 

As for the definitions for the intermediate set $\mathbb{I}$ and read-then-writer $\curvearrowright$ in the row-wise dependency analysis, we use the same definition in Definition~\ref{def-intermediate-set} and Definition~\ref{def-read-then-writer} from the column-wise dependency analysis. As for the row-wise dependency rule, we use Preposition~\ref{prop-row-dependency-rule1} and \ref{prop-dependency-rule2}. As for the property of the $\mathbb{I}$ set, we use the same Preposition~\ref{prop-i-set}, \ref{prop-read-then-writer1}, and \ref{prop-read-then-writer2} from the column-wise dependency analysis.

\begin{theorem}
For a retroactive operation for adding or removing the target query $Q'_\tau$ based on the row-wise query dependency analysis, it is sufficient to do the following: \textit{(i)} rollback the queries that belong to $\mathbb{I}$; \textit{(ii)} either execute $Q'_\tau$ (in case of retroactively adding $Q'_\tau$) or roll back $Q'_\tau$ (in case of retroactively removing $Q'_\tau$); \textit{(iii)} replay all queries in $\mathbb{I}$.
\label{theorem-row-dependency}
\end{theorem}

\begin{proof}
This proof follows the identical structure as the one used in the column-wise query dependency analysis in Theorem~\ref{theorem-column-dependency}, while the difference is the granularity of data records. The column-wise dependency analysis divides data records vertically as columns in tables (by using column names that each query accesses), whereas the row-wise dependency analysis divides data records horizontally as rows in tables (by using each table's RI column's values as rows that each query accesses). Regardless of whether the granularity of the data records is vertical or horizontal, the identical proof of induction holds. In particular, we replace
Proposition~\ref{prop-dependency-rule1} by Proposition~\ref{prop-row-dependency-rule1} to modify granularity of data records from columns to rows, while keeping the other propositions the same. 
\end{proof}

\subsection{Final Query Dependency Analysis}
\label{appendix:dependency-analysis} 

To avoid confusion, we rename the $\mathbb{I}$ set used in ~\autoref{theorem-column-dependency} (i.e., the column-wise dependency analysis) as $\mathbb{I}_c$, and the $\mathbb{I}$ set used in ~\autoref{theorem-row-dependency} (i.e., the column-wise dependency analysis) as $\mathbb{I}_r$.

Combining ~\autoref{theorem-column-dependency} and ~\autoref{theorem-row-dependency}, we finally reach ~\autoref{theorem-dependency}:

\begin{theorem}
For a retroactive operation for adding or removing the target query $Q'_\tau$, it is sufficient to do the following: \textit{(i)} rollback the queries that belong to $\mathbb{I}_c \cap \mathbb{I}_r$; \textit{(ii)} either execute $Q'_\tau$ (in case of retroactively adding $Q'_\tau$) or roll back $Q'_\tau$ (in case of retroactively removing $Q'_\tau$); \textit{(iii)} replay all queries in $\mathbb{I}_c \cap \mathbb{I}_r$.
\label{theorem-dependency}
\end{theorem}

\begin{proof}

For a retroactive operation of adding/removing the target query $Q'_\tau$, ~\autoref{theorem-column-dependency} states that it is sufficient to roll back and replay only those queries in $\mathbb{I}_c$, and ~\autoref{theorem-row-dependency} states that it is sufficient to roll back and replay only those queries in $\mathbb{I}_r$. Since both theorems state the condition of sufficiency, combining the conditions of these two theorems gives us the corollary that for a retroactive operation of adding/removing the target query $Q'_\tau$, it is sufficient to roll back and replay only those queries in $\mathbb{I}_c \cap \mathbb{I}_r$. 

\end{proof}

\newpage


\subsection{Collision Rate of \ultraverse's Table Hashes}
\label{appendix:hashing}

\ultraverse computes a table's hash by hashing its each row with a collision-resistant hash function and summing them up. 
By assuming that the collision-resistant hash function's output is uniformly distributed in $[0, p - 1]$, 
we will prove that given two tables $T_1$ and $T_2$, 
\ultraverse's Hash-jumper's hash collision rate is upper-bounded by  
{$\frac{1}{p}$} (i.e., with a probability no more than $\frac{1}{p}$ producing the same hash value for $T_1$ and $T_2$, when $T_2 \neq T_1$). 

Suppose the Hash-jumper outputs a hash value $h \in [0, p - 1]$ for $T_1$. Without loss of generality, we assume $T_2$ has $n$ 
rows. Then we prove by induction.

$ $


\noindent \textbf{Case} $\mathbf{n = 1}$: 
Because the collision-resistant hash function's output is uniformly distributed in $[0, p - 1]$, it is easy to see that the probability that the Hash-jumper outputs $h$ for any $T_2$ is $\frac{1}{p}$. 

$ $


\noindent \textbf{Case} $\mathbf{n = k}$: For each row of $T_2$, let $x_i$ denote the collision-resistant hash function's output. Then there are $p^k$ possibilities of $(x_1, x_2, ..., x_k)$ for the $k$ rows. Because the collision-resistant hash function's output is uniformly distributed in $[0, p - 1]$, all these possibilities have the same probability $\frac{1}{p^k}$. Consider the output of the Hash-jumper. For any $h_k \in [0, p - 1]$, we assume there are $p^{k-1}$ possibilities of $(x_1, x_2, ..., x_k)$ such that $\sum_i^k x_i = h_k$ $\mathbf{mod}$ $p$. This holds for $k = 1$, as we have seen for \textbf{Case} ${\mathbf{n = 1}}$. 

$ $


\noindent \textbf{Case} $\mathbf{n = k + 1}$: There are $p^{k + 1}$ possibilities of $(x_1, x_2, ..., x_{k+1})$ for the $(k+1)$ rows. Because the collision-resistant hash function's output is uniformly distributed in $[0, p - 1]$, all these possibilities have the same probability $\frac{1}{p^{k+1}}$. For any $(x_1, x_2, ..., x_k)$ such that $\sum_i^k x_i = h_k$ $\mathbf{mod}$ $p$, there exists exactly one $x_{k+1}$ such that $h_k + x_{k+1} = h$ $\mathbf{mod}$ $p$. By the assumption in \textbf{Case} $\mathbf{n = k}$, for each $h_k \in [0, p - 1]$, there are exactly $p^{k-1}$ possibilities of $(x_1, x_2, ..., x_{k+1})$ such that the output of the Hash-jumper is $h$. Because there are $p$ possible $h_k$ values in $[0, p - 1]$, there are $p^k$ possibilities of $(x_1, x_2, ..., x_{k+1})$ such that the output of the Hash-jumper is $h$. Therefore, the probability that that the Hash-jumper outputs $h$ for a table $T_2$ of $(k+1)$ rows is $\frac{p^k}{p^{k+1}} = \frac{1}{p}$. 

$ $

Because the above probability is independent of the number of rows $n$, for any $T_2$, the Hash-jumper output $h$ with a probability of $\frac{1}{p}$. Therefore, the hash collision rate is upper-bounded by $\frac{1}{p}$ when $T_2 \neq T_1$.

\textbf{False Positives:} From the security perspective, there is yet a non-zero chance that a malicious user fabricates two row hash values $(x_1', x_2')$ such that $(x_1' + x_2')$ $\mathbf{mod}$ $p = (x_1 + x_2)$ $\mathbf{mod}$ $p$ and tricks the database server into believing in a false positive on a table hash hit. To address this, whenever a table hash hit is found, Hash-jumper optionally makes a literal table comparison between two table versions at the same commit time (one evolved during the replay; the other newly rolled back from the original database to this same point in time) and verifies if they really match. If the literal table comparison returns a true positive before finishing to replay rest of the queries, \ultraverse still ends up reducing its replay time. 

\textbf{False Negatives:} A subtlety occurs when a query uses the \texttt{LIMIT} keyword without \texttt{ORDER} \texttt{BY}, because each replay of this query may return different row(s) of a table in a non-deterministic manner, which may lead Hash-jumper to making a false negative decision and missing the opportunity of a legit hash-jump. However, note that missing the opportunity of optimization does not affect the correctness of retroactive operations.


\end{document}